\documentclass[10pt,reqno]{amsart}
\usepackage[utf8x]{inputenc}
\usepackage[english]{babel}
\usepackage{amsmath,amsfonts,amssymb}
\usepackage[T1]{fontenc}
\usepackage{pxfonts}
\usepackage{lmodern}

\usepackage[top=3.5cm,bottom=3.5cm,left=3.3cm,right=3.3cm]{geometry}
\usepackage[dvipsnames]{xcolor}
\usepackage[hyperindex=true,frenchlinks=true,colorlinks=true,
    citecolor=ForestGreen,linkcolor=BrickRed,urlcolor=LimeGreen,
    linktocpage]{hyperref}

\usepackage{cite}
\usepackage{float}
\usepackage{url}
\usepackage{bbm}
\usepackage{stmaryrd}
\usepackage{enumitem}

\usepackage{tikz}
\usepackage{tkz-graph}
\usepackage{tkz-berge}
\usetikzlibrary{automata}
\usetikzlibrary{shadows}
\usetikzlibrary{snakes}
\usetikzlibrary{arrows}
\usetikzlibrary{shapes}
\usetikzlibrary{decorations.pathmorphing}
\usetikzlibrary{fit}

\newtheorem{lemma}{Lemma}
\newtheorem{example}{Example}
\newtheorem{remark}{Remark}

\newtheorem{proposition}{Proposition}
\newtheorem{theorem}{Theorem}
\newtheorem{claim}{Claim}
\newtheorem{corollary}{Corollary}

\numberwithin{equation}{section}
\renewcommand{\leq}{\leqslant}
\renewcommand{\geq}{\geqslant}

\newcommand{\tright}{{\rhd\hskip-5pt\rightarrow}}

\newcommand{\Shift}[2]{{\displaystyle\mathop\circleright^{#1}(#2)}}

\newcommand{\ShiftDiam}[2]{{\displaystyle\mathop\Diamondright^{#1}(#2)}}
\newcommand{\ShiftDiamnk}[1]{{\displaystyle\mathop\Diamondright^{#1}}}

\newcommand{\ShiftBoxnk}[1]{{\displaystyle\mathop\boxRight^{#1}}}
\newcommand{\Shiftdnk}[1]{{\displaystyle\mathop\circleddotright^{#1}}}

\newcommand{\ShiftTnk}[1]{{\displaystyle\mathop\tright^{#1}}}
\newcommand{\Shiftnk}[1]{{\displaystyle\mathop\circleright^{#1}}}

\newcommand{\ShiftDiamD}[2]{{\displaystyle\mathop\Diamonddotright^{#1}(#2)}}
\newcommand{\ShiftDiamDnk}[1]{{\displaystyle\mathop\Diamonddotright^{#1}}}
\newcommand{\goth}{\mathfrak}
\newcommand{\N}{{\mathbb N}}
\newcommand{\Had}{{\mathbb H}}

\newcommand{\ARef}{\mathrm{ARef}}

\newcommand{\Z}{{\mathbb Z}}
\newcommand{\RAS}{\mathrm{RAS}}
\newcommand{\ARAS}{\mathrm{ARAS}}
\newcommand{\pp}{\mathbf{p}}
\newcommand{\Free}{\mathrm{Free}}
\newcommand{\Hom}{\mathrm{Hom}}
\newcommand{\OP}{\mathrm{OP}}
\newcommand{\QOSET}{\mathrm{QOSet}}

\tikzstyle{every picture}=[>=stealth',shorten >=1pt,node distance=1.44cm,
bend angle=45,initial text=,
every state/.style={inner sep=0.75mm, minimum size=1mm},font=\scriptsize]
\tikzstyle{Injection} = [black!100,draw,>->]
\tikzstyle{Surjection} = [black!100,draw,->>]

\title{Operads, quasiorders, and regular languages}
\date{\today}
\author{Samuele Giraudo}
\address{Laboratoire d'Informatique Gaspard-Monge, Université
    Paris-Est Marne-la-Vallée, 5 boulevard Descartes, Champs-sur-Marne,
    77454 Marne-la-Vallée cedex 2, France.}
\email{samuele.giraudo@univ-mlv.fr}

\author{Jean-Gabriel Luque}
\address{Laboratoire LITIS - EA 4108
    Université de Rouen. Avenue de l'Université - BP 8
    76801 Saint-Étienne-du-Rouvray Cedex.}
\email{jean-gabriel.luque@univ-rouen.fr}

\author{Ludovic Mignot}
\address{Laboratoire LITIS - EA 4108
    Université de Rouen. Avenue de l'Université - BP 8
    76801 Saint-Étienne-du-Rouvray Cedex.}
\email{ludovic.mignot@univ-rouen.fr}

\author{Florent Nicart}
\address{Laboratoire LITIS - EA 4108
    Université de Rouen. Avenue de l'Université - BP 8
    76801 Saint-Étienne-du-Rouvray Cedex.}
\email{florent.nicart@univ-rouen.fr}

\begin{document}

\begin{abstract}
We generalize the construction of multi-tildes in the aim to provide
double multi-tilde operators for regular languages. We show that the underlying
algebraic structure involves the action of some operads. An operad is an
algebraic structure that mimics the composition of the functions. The
involved operads are described in terms of combinatorial objects. These
operads are obtained from more primitive objects, namely  precompositions,
whose algebraic counter-parts are investigated. One of these operads acts
faithfully on languages in the sense that two different operators act in
two different ways.
\end{abstract}

\maketitle

\tableofcontents

\section*{Introduction}
Following the Chomsky-Sch\"utzenberger hierarchy~\cite{CH56}, regular
languages are defined  to be the formal languages that are generated by
Type-3 grammars (also called regular grammars). These particular languages
have been studied from several years since they have many applications in
several areas such as pattern matching, compilation, verification, and
bioinformatics. Their generalization as rational series links them to
various algebraic or combinatorial topics like enumeration (manipulation
of generating functions), rational approximation (for instance Pade
approximation), representation theory (modules viewed as automata), and
combinatorial optimization ($(\max,+)$-automata).
\smallskip

One of the main specificity of regular languages is that they can be
represented by various tools: regular grammars, automata, regular
expressions, \emph{etc.} Whilst they can be represented by both automata
and regular expressions~\cite{Kle56}, these two tools are not equivalent.
Indeed, Ehrenfeucht and Zeiger~\cite{EZ76} showed a one parameter family
of automata whose shortest equivalent regular expressions have a width
exponentially growing with the numbers of states. Note that it is possible
to compute an automaton from  a regular expression $E$ such that the number
of its states is a linear function  of the alphabet width (\emph{i.e.},
the number of occurrences of alphabet symbols)
of~$E$~\cite{Ant96,CZ01a,Glu61,MY60}.
\smallskip

In the aim to increase expressiveness of regular expressions for a bounded
length, Caron~\emph{et al.}~\cite{CCM11a} introduced the so-called
multi-tilde operators and applied these to represent finite languages.
Investigating the equivalence of two multi-tilde expressions, they define
a natural notion of composition which endows the set of multi-tilde
operators with a structure of operad. This structure has been investigated
in~\cite{LMN12}.
\smallskip

Originating from the algebraic topology~\cite{May72,BV73}, operad theory
has been developed as a field of abstract algebra concerned by prototypical
algebras that model classical properties such as commutativity and
associativity~\cite{LV10}. Generally defined in terms of categories, this
notion can be naturally applied to computer science. Indeed, an operad is
just a set of operations, each one having exactly one output and a fixed
finite number of inputs, endowed with the composition operation. An
operad can then model the compositions of functions occurring during the
execution of a program. In terms of theoretical computer science, this
can be represented by trees with branching rules. The whole point of the
operads in the context of the computer science is that this allows to use
different tools and concepts from algebra (such as morphisms,
quotients, substructures, generating sets).
\smallskip

In order to illustrate this point of view, let us recall the main results
of our previous paper~\cite{LMN12}. In this paper, we first showed that
the set of multi-tilde operators has a structure of operad. We used the
concept of morphism in the aim to choose the operad allowing us to describe
in the simplest way a given operation or a property. For instance, the
original definition of the action of the multi-tildes on languages is
rather complicated. But, \emph{via} an intermediate operad based on set
of boolean vectors, the action was described in a more natural way. In
the same way, the equivalence problem is clearer when asked in a operad
based on antisymmetric and reflexive relations which is isomorphic to the
operad of multi-tildes: two operators are equivalent if and only if they
have the same transitive closure. The transitive closure being compatible
with the composition, we defined an operad based on partial ordered sets
as a quotient of the previous operad and we showed that this representation
is optimal in the sense that two different operators act in two different
ways on languages. This not only helps to clarify constructions but also
to ask new questions. For instance, how many different ways do $n$-ary
multi-tildes act on languages? Precisely, the answer is the number of
posets on $\{1,\dots,n+1\}$ that are compatible with the natural order
on integers.
\smallskip

The goal of this paper is to generalize this construction to regular
languages. We investigate several operads (based on double multi-tildes,
antireflexive relations or quasiorders) allowing to represent a regular
language as an $n$-ary operator acting on a $n$-tuple of symbols
$(\alpha_1,\dots,\alpha_n)$ where the $\alpha_i$ are symbols or
$\emptyset$. These operators generalize the multi-tildes and the
investigated properties involve their underlying operads. Such a generalization  induces the definition of new parametrized operators that allow to increase the number of regular languages   denoted by an expression with a fixed alphabetical width \emph{i.e.}, the number of occurrences of its symbols. 
One of the main properties of such a family of operators is that the expressions can be easily translated in terms of $\varepsilon$-automata.
This paper is the first step of this process: multi-tildes were shown to be able to replace the operators of sum and catenation; this work shows that they may replace the Kleene star too. The notion of precomposition is the keystone of this modus operandi; using functors and category theory, the next step is to link the notion of operad and the conversions between expressions and automata.

\smallskip

This paper is organized as follows. First we recall in
Section~\ref{sec:comb op lang theor} several notions concerning operad
theory and multi-tilde operations. In Section~\ref{sec3}, we observe that
many of the operads involved in~\cite{LMN12} and in this paper have some
common properties. More precisely, they can be described completely by
means of {\em shifting} operations. This leads to the definition of the
category of precompositions together with a functor to the category of
operads. We also define and investigate the notion of quotients of
precompositions. These structures serve as model for the operads defined
in the sequel. To illustrate how to use these tools, we revisit
in Section~\ref{sec4} the operads defined in~\cite{LMN12} and describe
them in terms of precompositions. In Section~\ref{sec5}, we define the
double multi-tilde operad $\mathcal{DT}$ as the graded tensor square of
the multi-tilde operad. We construct also an isomorphic operad $\ARef$
based on antireflexive relations and a quotient  based on quasiorders
$\QOSET$. In Section~\ref{sec6}, we describe the action of the
operads on the languages. In particular, we show that any regular language
can be written as $\pp_n(\alpha_1,\dots,\alpha_n)$ where the
$\alpha_i$ are letters or $\emptyset$ and $\pp_n$ is an $n$-ary
operation belonging to $\ARef$, $\mathcal{DT}$, or $\QOSET$.
Finally, we prove that the action of $\QOSET$ on regular
languages is faithful, that is two different operators act in two
different ways.
\smallskip

The operad studied in this paper fit into the following diagram
\begin{equation}
    \begin{split}
    \begin{tikzpicture}[xscale=1,yscale=.6]
        \node(ARef)at(0,0){
            \begin{math}\mathrm{ARef} \simeq
            \Had(\mathrm{ARAS}, \mathrm{ARAS}) \simeq
            \mathcal{DT}\end{math}};
        \node(QOSet)at(-2,-3){\begin{math}\QOSET\end{math}};
        \node(ARAS)at(2,-3){
            \begin{math}\mathrm{ARAS} \simeq
            \mathrm{RAS} \simeq
            \mathcal{T}\end{math}};
        \node(POSet)at(0,-6){\begin{math}\mathrm{POSet}\end{math}};
        \draw[Surjection](ARef)--(QOSet);
        \draw[Surjection](ARAS)--(POSet);
        \draw[Injection](ARAS)--(ARef);
        \draw[Injection](POSet)--(QOSet);
    \end{tikzpicture}
    \end{split}
\end{equation}
where arrows $\rightarrowtail$ (resp. $\twoheadrightarrow$) are
injective (resp. surjective) morphisms of operads.
\bigskip

{\em Acknowledgements}. The authors would like to thank the referee for 
his comments improving the quality of the paper.
\medskip

\section{Some combinatorial operators in language theory}%
\label{sec:comb op lang theor}

We recall here some basic notions about the theory of operads and set
our notations for the sequel of the paper. In particular, we recall
what are operads, free operads, and modules over an operad. We conclude
this section by presenting the operad of multi-tildes introduced
in~\cite{LMN12}.

\subsection{Nonsymmetric operads}
Since we shall consider in this paper only nonsymmetric operads, we
shall call these simply {\em operads}. Operads are algebraic graded
structures which mimic the composition of $n$-ary operators. Let us recall
the main definitions and properties. Let
$\goth P=\bigsqcup_{n \geq 1}\goth P_n$ be a graded set
($\bigsqcup$ means that the sets are disjoint); the elements of $\goth P_n$
are called {\em $n$-ary operators}. The set $\goth P$ is endowed with maps
\begin{equation}
    \circ_i:\goth P_n\times \goth P_m\rightarrow \goth P_{n+m-1},
\end{equation}
where $1 \leq i \leq n$, called {\em partial compositions} and satisfying
for all $\pp_1 \in \goth P_n$, $\pp_2 \in \goth P_m$, and $\pp_3 \in \goth P_p$
the two following rules.
\begin{enumerate}
    \item {\it Associativity 1}: if $1\leq i<j\leq n$, then
    \begin{equation}\label{Associativity1}
    (\pp_1\circ_i\pp_2)\circ_{j+m-1} \pp_3
    =(\pp_1\circ_j \pp_3)\circ_i\pp_2.
    \end{equation}
    \item {\it Associativity 2}: if $j\leq m$, then
\begin{equation}\label{Associativity2}
(\pp_1\circ_i\pp_2)\circ_{i+j-1}\pp_3
    =\pp_1\circ_i(\pp_2\circ_j\pp_3).
    \end{equation}
\end{enumerate}
Moreover, in an operad, there is a special element
$\mathbf1 \in \goth P_1$ called {\em unit} and satisfying for
all $\pp \in \goth P_n$ the following rule.
\begin{enumerate}[resume]
    \item {\it Unitality relation}: if $1 \leq i \leq n$, then
    \begin{equation}
        \pp \circ_i \mathbf1 = \pp = \mathbf1 \circ_1 \pp.
    \end{equation}
\end{enumerate}
The reader could refer to~\cite{LV10,Markl} for a complete description
of the structures.
\medskip

Consider two operads $(\goth P,\circ,\mathbf1)$ and
$(\goth P',\circ',\mathbf1')$. A {\em morphism of operads} is a graded
map $\phi:\goth P\rightarrow \goth P'$ such that
$\phi(\mathbf1) = \mathbf1'$ and
\begin{equation}
    \phi(\pp_1\circ_i\pp_2)=\phi(\pp_1)\circ'_i\phi(\pp_2)
\end{equation}
for all $\pp_1\in\goth P_n$, $\pp_2\in\goth P_m$ and $1\leq i\leq n$. If
$\goth Q\subset \goth P$, the {\em suboperad} of $\goth P$ generated by
$\goth Q$ is the smallest subset of $\goth P$ containing $\goth Q$ and
$\mathbf 1$ which is stable by composition.
\medskip

Let $\goth G = \bigsqcup_{n\geq 1}\goth G_n$ be a graded set. The set
$\Free(\goth G)_n$ is the set of planar rooted trees with $n$ leaves and
where any internal node with $n$ children is labeled on $\goth G_n$. The
{\em free operad} on $\goth G$ is obtained by endowing the set
$\Free(\goth G) = \bigsqcup_{n \geq 1} \Free(\goth G)_n$ with the partial
compositions $\circ_i$ where $\pp_1 \circ_i \pp_2$ is the tree obtained
by grafting the $i$th leaf of $\pp_1$ on the root of $\pp_2$. Observe
that $\Free(\goth G)$ contains a copy of $\goth G$ which is the set of
the trees with only one internal node (the root) labeled on $\goth G$;
for simplicity we will identify it with $\goth G$. Moreover,
$\Free(\goth G)$ is clearly generated by $\goth G$. The universality
means that for any map $\varphi: \goth G \rightarrow \goth P$ there
exists a unique operad morphism $\phi: \Free(\goth G) \rightarrow\goth P$
such that $\phi(\mathbf g) = \varphi(\mathbf g)$ for each
$\mathbf g\in \goth G$.
\medskip

A graded equivalence relation $\equiv$ on $\goth P$ is an
{\em operad congruence} if for all $\pp_1,\pp_2,\pp'_1,\pp'_2 \in \goth P$,
$\pp_1\equiv\pp'_1$, and $\pp_2\equiv\pp'_2$ imply
$\pp_1\circ_i\pp_2\equiv \pp'_1\circ_i\pp'_2$. The set $\goth P/_\equiv$
is then naturally endowed with a structure of operad, called
{\em quotient operad}. Note that if $\phi:\goth P\rightarrow \goth P'$
is a surjective morphism of operads then the equivalence defined by
$\pp_1\equiv\pp_1'$ if and only if $\phi(\pp_1)=\phi(\pp_1)$ is an operad
congruence.
\medskip

The {\em Hadamard product} $\Had(\goth P,\goth P')$ of two operads
$(\goth P,\circ)$ and $(\goth P',\circ')$ is the operad defined as
follows. The elements of arity $n$ of $\Had(\goth P,\goth P')$ are pairs
$(\pp, \pp')$ where $\pp \in \goth P_n$ and $\pp' \in \goth P'_n$ and
its partial compositions are defined by
\begin{equation}
    (\pp_1,\pp'_1)\circ_i (\pp_2, \pp'_2):=
    (\pp_1\circ_i\pp_2, \pp'_1\circ_i'\pp'_2),
\end{equation}
for all $(\pp_1,\pp'_1) \in \Had(\goth P,\goth P')_n$,
$(\pp_2, \pp'_2) \in \Had(\goth P,\goth P')_m$, and $1 \leq i \leq n$.
\medskip

Consider a set $\mathbf S$ together with a left action of an operad
$\goth P$. That is, each $\pp \in\goth P_n$ defines a map
\begin{equation}
    \pp: \mathbf S^n \rightarrow\mathbf S.
\end{equation}
We say that $\mathbf S$ is a $\goth P$-{\em module} if the action of
$\goth P$ is compatible with the composition in the following sense. For
any $\pp_1 \in \goth P_n$, $\pp_2 \in \goth P_m$, $1\leq i\leq n$, and
$s_1,\dots, s_{n+m-1} \in\mathbf S$, one has
\begin{equation}
    \pp_1(s_1,\dots,s_{i-1},
    \pp_2(s_i,\dots,s_{i+m-1}),
    s_{i+m},\dots, s_{n+m-1})
    = (\pp_1 \circ_i \pp_2)(s_1,\dots, s_{n+m-1}).
\end{equation}
Furthermore, if for each $n\geq1$ and $\pp \neq \pp'\in\goth P_n$ there
exist $s_1,\dots,s_n\in \mathbf S$ such that
$\pp(s_1,\dots,s_n)\neq \pp'(s_1,\dots,s_n)$, we say that the module
$\mathbf S$ is {\em faithful}.
\medskip

\subsection{Multi-tildes and related operads}
In~\cite{LMN12}, we have defined several operads. Let us recall briefly
the main constructions. First we defined the {\em operad of multi-tildes}
$\mathcal T=\bigsqcup_n \mathcal T_n$. A multi-tilde of $\mathcal T_n$
is a subset of $\{(x,y):1\leq x\leq y\leq n\}\subset \N^2$. Note that $\bigsqcup_n$
means that sets belonging in two different graded components
$\mathcal T_n$ and $\mathcal T_m$ are considered as different operators.
Let $i\geq 0$ and $n\geq 1$,  for any pair $(x,y)$ of positive integers, we define
\begin{equation}
    \Shift{i,n}{x,y} :=
    \begin{cases}
        (x,y)&\mbox{ if } y \leq i - 1,\\
        (x,y+n-1)&\mbox{ if }x\leq i\leq y,\\
        (x+n-1,x+n-1)&\mbox{ otherwise}.
    \end{cases}
\end{equation}
The actions of the  operators are extended to sets $E$ of pairs of
positive integers by
\begin{equation}
    \Shift{i,n}E:=\{\Shift{i,n}{x,y}:(x,y)\in E\}.
\end{equation}
\medskip

We have shown the following result:
\begin{theorem}[\!\cite{LMN12}]\label{MToperade}
    The set $\mathcal T$ endowed with the partial compositions
    \begin{equation}
        \circ_i:\left\{
        \begin{array}{rcl}
        \mathcal T_n\times\mathcal T_m&\rightarrow&\mathcal T_{n+m-1}\\
        T_1\circ_i T_2&\mapsto&\Shift{i,m}{T_1}\cup\Shift{0,i}{T_2},
        \end{array}
        \right.
    \end{equation}
    is an operad.
\end{theorem}
\medskip

We have also defined the operators
\begin{equation}
    \ShiftDiam{i,n}x:=
    \begin{cases}
        x&\mbox{if } x\leq i,\\
        x+n-1&\mbox{otherwise},
    \end{cases}
\end{equation}
and have extended these respectively to pairs and sets of pairs by
\begin{equation}
\ShiftDiam{i,n}{x,y}=(\ShiftDiam{i,n}x, \ShiftDiam{i,n}y)
\end{equation}
and
\begin{equation}
\ShiftDiam{i,n}E=\{\ShiftDiam{i,n}{x,y}:(x,y)\in E\},
\end{equation}
where $x$ and $y$ are positive integers and $E$ is a set of pairs
of positive integers.
The operad $(\mathcal T,\circ)$ is isomorphic to another operad
$(\RAS,\Diamond)$ whose underlying set is the set $\RAS=\bigsqcup_n \RAS_n$
where $\RAS_n$ denotes the set of Reflexive and Antisymmetric Subrelations
of the natural order $\leq$ on $\{1,\dots,n+1\}$. The partial compositions
of $\RAS$ are defined by
\begin{equation}
    R_1\Diamond_i R_2 :=\ShiftDiam{i,m}{R_1}\cup\ShiftDiam{0,i}{R_2},
\end{equation}
if $R_1\in\RAS_n$ and $R_2\in\RAS_m$. The isomorphism of operads
$\phi : \mathcal T \to \RAS$ satisfies, for all $T\in\mathcal T_n$,
\begin{equation}
    \phi(T) = \{(x,y+1):(x,y)\in T\}\cup\{(x,x):x\in\{1,\dots,n+1\}\}.
\end{equation}
See~\cite{LMN12} for more details.
\medskip

\section{Breaking operads}\label{sec3}
The objective of this section is to introduce new algebraic objects,
namely the {\em precompositions}. The precompositions are a kind of representation of a certain
monoid denoted by $\bigbox$ which can be described in terms of infinite matrices.
 We present here a functor from the
category of precompositions to the category of operads. We shall use
this functor in the sequel to reconstruct some already known operads
and to construct new ones.
\medskip

\subsection{Some monoids of infinite matrices}

We consider the set $\overline M_\infty$  of infinite matrices with a finite number of non-zero diagonals whose entries belong to the boolean semiring $\mathbb B$. A typical element $(a_{ij})_{i,j\in\Z}$ of $\overline M_\infty$ is a finite linear combination of elements $D_{k}^\lambda:=\sum_{i\in \Z}\lambda_iE_{i+k,i}$ where $\lambda=(\lambda_i)_{i\in\Z}$ and $E_{k,\ell}=\left(\delta_{i,k}\delta_{j,\ell}\right)_{i,j\in \Z}$ and $\delta_{i,j}=1$ if $i=j$ and $0$ otherwise is the Kronecker symbol \emph{i.e.}, the matrix $E_{k,\ell}$ has $1$ at the cell $(k,\ell)$ and $0$ elsewhere.
\medskip

Observing that
\begin{equation}\label{DD2D}
D_{k}^\lambda D^{\lambda'}_{k'}=\left(\sum_{i\in\Z}\lambda_iE_{i+k,i}\right)\left(\sum_{i\in\Z}\lambda'_iE_{i+k',i}\right)=
\sum_{i\in\Z}\lambda_{i+k'}\lambda'_iE_{i+k+k',i}=D_{k+k'}^{\lambda\bullet_{k'}\lambda'},
\end{equation}
with $\lambda\bullet_{k'}\lambda'=(\lambda_{i+k'}\lambda'_i)_{i\in\Z}$,
we deduce that $\overline M_\infty$ is stable by product. So
\begin{proposition}
The space $\overline M_\infty$ is an algebra.
\end{proposition}
\def\B{\mathbb B}
\def\C{\mathbb C}
\begin{remark}Notice that  when the entries belong to $\C$ instead of $\B$, the algebraic structure of $\overline M_\infty$ is very rich and has many connexions with the study of infinite Lie algebras (see \emph{e.g.} \cite{Kac}).
\end{remark}

Here, for our purpose, we consider only the structure of monoid; the unit of  $\overline M_\infty$ is  $Id:=D^{(\dots,1,1,1,\dots)}_1=\sum_{i\in\Z} E_{i,i}$. In particular, we define the submonoid $\overline P_\infty$ generated by the matrices \begin{equation}\mathfrak m_{i,n}=\sum_{j\leq i}E_{j,j}+\sum_{  i< j}E_{j+n-1,j}\end{equation} for each $i\in\Z$ and each $n>0$. With these notations we have
\begin{equation}\label{m2id}
\mathfrak m_{i,1}=Id,
\end{equation}
 for any $i$.
\medskip

Let $]-\infty,i]$ be the vector such that $]-\infty,i]_j=1$ if $j\leq i$ and $0$ otherwise, and $]i,\infty[=(1-]-\infty,i]_j)_{j\in\Z}$.  With these notations one has
\begin{equation}\label{m2D}
\mathfrak m_{i,n}=D_1^{]-\infty,i]}+D_{n}^{]i,\infty[}.
\end{equation}
\begin{lemma} The two following identities hold:
\begin{itemize}
\item if $i\leq j$ then
\begin{equation}\label{com_m}
\mathfrak m_{i,n}\mathfrak m_{j,n'}=\mathfrak m_{j+n-1,n'}\mathfrak m_{i,n},
\end{equation}
\item if $0\leq j< n'$ then
\begin{equation}\label{simp_m}
\mathfrak m_{i+j,n}\mathfrak m_{i,n'}=\mathfrak m_{i,n+n'-1}.
\end{equation}
\end{itemize}
\end{lemma}
\begin{proof}
From (\ref{m2D}), one has
\[\begin{array}{rcl}
\mathfrak m_{i,n}\mathfrak m_{j,n'}&=&
\left(D_0^{]-\infty,i[}+D_{-1}^{]i,\infty[}\right)\left(D_0^{]-\infty,j]}+D_{n'-1}^{]j,\infty[}\right)
\\&=&D_0^{]-\infty,i]}D_0^{]-\infty,j]}
+D_0^{]-\infty,i]}D_{n'-1}^{]j,\infty[}
+D_{n-1}^{[i,\infty]}D_0^{]-\infty,j]}+
D_{n-1}^{[i,\infty]}D_{n'-1}^{]j,\infty[}.\end{array}
\]

But  $i\leq j$ implies $]-\infty,i]\bullet_0]-\infty,j]=]-\infty,i]$, $]-\infty,i]\bullet_{n'-1}]j,\infty[=[\dots,0,0,0,\dots]$, $]i,\infty[\bullet_{0} ]-\infty,j]= ]-\infty,j+n-1]\bullet_{n-1} ]i,\infty[$ and $
]i,\infty[\bullet_{n'-1}]j,\infty[=]j,\infty[$.
Hence,  (\ref{DD2D}) implies
\[
\mathfrak m_{i,n}\mathfrak m_{j,n'}=D_0^{]-\infty,i]}+D_{n-1}^{]-\infty,j+n-1]\bullet_{n-1} ]i,\infty[}+
D_{n+n'-2}^{]j,\infty[}=\mathfrak m_{j+n-1,n'}\mathfrak m_{i,n}.\]
This proves formula (\ref{com_m}).

Now, from (\ref{m2D}), we obtain
\[\begin{array}{rcl}
\mathfrak m_{i+j,n}\mathfrak m_{i,n'}&=&(D_0^{]-\infty,i+j]}+D_{n-1}^{]i+j,\infty[})(D_0^{]-\infty,i]}+D_{n'-1}^{]i,\infty[})\\&=&D_0^{]-\infty,i+j]}D_0^{]-\infty,i]}+D_0^{]-\infty,i+j]}D_{n'-1}^{]i,\infty[}+
D_{n-1}^{]i+j,\infty[}D_0^{]-\infty,i]}+D_{n-1}^{]i+j,\infty[}D_{n'-1}^{]i,\infty[}
\end{array}\]
But, since $0\leq j< n'$, one has $]-\infty,i+j]\bullet_0]-\infty,i]=]-\infty,i]$, $]-\infty,i+j]\bullet_{n'-1}]i,\infty[=]i+j,\infty[\bullet_0]-\infty,i]=[\dots,0,0,0,\dots,]$  and $]i+j,\infty[\bullet_{n'-1}]i,\infty[=]i,\infty[$.
Hence,  (\ref{DD2D}) allows us to recover (\ref{simp_m}):
\[
\mathfrak m_{i+j,n}\mathfrak m_{i,n'}=D_0^{]-\infty,i]}+D_{n+n'-2}^{]i,\infty[}=\mathfrak m_{i,n+n'-1}.
\]
\end{proof}

\begin{proposition} (Presentation of $\overline P_\infty$)\\
The monoid $\overline P_\infty$ is isomorphic to the monoid $\overline\bigbox$ generated by the symbols $\{\ShiftBoxnk{i,n}:i\in\Z, n\geq 1\}$ and the relations
\begin{eqnarray}
\ShiftBoxnk{i,1}=\mathbf 1_{\overline\bigbox}\ \mbox{ for any }i\in\Z.\\
\ShiftBoxnk{i,n}\ShiftBoxnk{j,n'} = \ShiftBoxnk{j+n-1,n'}\ShiftBoxnk{i,n}\ \mbox{ if }i\leq j\label{cond2},\\
\ShiftBoxnk{i+j,n}\ShiftBoxnk{i,n'}=\ShiftBoxnk{i,n+n'-1}\ \mbox{ if }0\leq j<n'\label{cond3}.
\end{eqnarray}
\end{proposition}
\begin{proof}
First let us prove that the map $\varphi$ sending $\ShiftBoxnk{i,n}$ to $\mathfrak m_{i,n}$ can be extended as a morphism of monoids from $\overline\bigbox$ to $\overline P_\infty$. It suffices to show that
$\varphi(\ShiftBoxnk{i,1})=Id$ for any $i\in\Z$, $\varphi(\ShiftBoxnk{i,n})\varphi(\ShiftBoxnk{j,n'})= \varphi(\ShiftBoxnk{j+n-1,n'})\varphi(\ShiftBoxnk{i,n})$ if $i\leq j$ and $\varphi(\ShiftBoxnk{i+j,n})\varphi(\ShiftBoxnk{i,n'})=\varphi(\ShiftBoxnk{i,n+n'-1})$ when $0\leq j<n'$. These equalities are respectively the consequences of (\ref{m2id}), (\ref{com_m}), and (\ref{simp_m}). Hence, $\varphi$ is extended to an into morphism of monoids, called also $\varphi$. It remains to prove that $\varphi$ is into.

Using (\ref{cond2}) and (\ref{cond3}), any element of $\overline\bigbox$ can be written as
\[
\ShiftBoxnk{i_0,n_0}\cdots \ShiftBoxnk{i_\ell,n_\ell}
\]
for some $i_0<\cdots<n_\ell$ and $n_0,\dots,n_\ell>1$.
Therefore, since  $\varphi$ is a morphism, any element of $\overline P_\infty$ can be written into the form $ \mathfrak m_{i_0,n_0}\cdots \mathfrak m_{i_\ell,n_\ell}$.
Furthermore, we observe that
\[\begin{array}{rcl}
 \mathfrak m_{i_0,n_0}\cdots \mathfrak m_{i_\ell,n_\ell}&=&\displaystyle
\sum_{j\leq i_0}E_{j,j}+\sum_{i_0<j\leq i_1}E_{j+n_0-1,j}+\sum_{i_1< j\leq i_2}E_{j+n_0+n_1-1,j}
+\cdots\\&&\displaystyle+ \sum_{i_{\ell-1}< j\leq i_\ell}E_{j+n_0+\cdots+n_{\ell-1}-\ell,j}+\sum_{i_\ell< j}E_{j+n_0+\cdots+n_\ell-\ell-1,j}.\end{array}
\]
As a consequence, it is easy to check that the factorization is unique .
Hence, since the image of  $\ShiftBoxnk{i_0,n_0}\cdots \ShiftBoxnk{i_\ell,n_\ell}$ by $\varphi$ is $\mathfrak m_{i_0,n_0}\cdots \mathfrak m_{i_\ell,n_\ell}$, each element of $\overline P_\infty$ admits a unique preimage and so $\varphi$ is injective.
It follows that the $\varphi$ is an isomorphism.
\end{proof}

Now, let us consider the algebra $M_\infty$ of the matrices $(a_{ij})_{i,j\in\N\setminus\{0\}}$ with a finite number of non-zero diagonals. Let $\pi: \overline M_\infty\longrightarrow M_\infty$ be the projection sending
$(a_{ij})_{i,j\in\Z}$ to $(a_{ij})_{i,j\in\N\setminus\{0\}}$ and $\mathfrak m_{i,n}=\pi(\mathfrak n_{i,n})=\sum_{0<j\leq i}E_{j,j}+\sum_{j > i,j>0}E_{j+n-1,j}$. Remarking that $\pi(\mathfrak m_{i,n}+\mathfrak m_{j,n'})=\mathfrak n_{i,n}+\mathfrak n_{j,n'}$, the restriction of $\pi$ to $\overline P_\infty$ is a morphism of monoid from $\overline P_\infty$ to the submonoid $P_\infty$ of $M_\infty$ generated by the matrices $\mathfrak n_{i,n}$. We notice also that $\mathfrak n_{-i,n}=\sum_{j\geq 1}E_{j+n-1,j}=\mathfrak n_{0,n}$ for any $i>0$.
Furthermore, we have
\begin{proposition}
The monoid $P_\infty$ is isomorphic to the quotient $\bigbox$ of the monoid $\overline \bigbox$ by the relations $\ShiftBoxnk{-i,n}=\ShiftBoxnk{0,n}$ for any $i>0$. That is the monoid defined by generators
$\{\ShiftBoxnk{i,n}:i\in\Z, n \geq 1\}$ and relations:
\begin{eqnarray}
    \ShiftBoxnk{i,n}=\ShiftBoxnk{0,n}
        & \mbox{ for any }i<0,\label{precompeq1}\\
    \ShiftBoxnk{i,1}=\ShiftBoxnk{0,1}=\mathbf 1_{\bigbox}
        & \mbox{ for any }i, \label{precompeq2}\\
    \ShiftBoxnk{i,n}\ShiftBoxnk{j,m} = \ShiftBoxnk{j+n-1,m}\ShiftBoxnk{i,n}
        & \mbox{ if }i\leq j\mbox{ or }i,j\leq 0,\label{precompeq3}\\
    \ShiftBoxnk{i+j,n}\ShiftBoxnk{i,m}=\ShiftBoxnk{i,n+m-1}
        & \mbox{ if }0\leq j<m.\label{precompeq4}
\end{eqnarray}

\end{proposition}

\begin{remark}\label{PK}
We have explained the construction when the entries are taken in $\mathbb B$. One can make a similar construction for any semiring $\mathbb K$ and obtain monoids $\overline{P}_\infty(\mathbb K)$ and
$P_\infty(\mathbb K)$. In all the case, the monoid $\overline{P}_\infty(\mathbb K)$ is isomorphic to $\overline{P}_\infty$ and the monoid $P_\infty(\mathbb K)$ is isomorphic to $P_\infty$.
\end{remark}
\subsection{Precompositions}
Let $(\mathcal S,\oplus)$ be a commutative monoid endowed with a filtration $\mathcal S=\bigcup_{n\geq 1}\mathcal S_n$  with
\begin{equation}
    \mathcal S_1\subset \mathcal S_2 \subset \cdots \subset \mathcal S_n
    \subset \cdots
\end{equation}
and such that  each $\mathcal S_n$ is a submonoid is $\mathcal S$. We will denote by $\mathbf 0_{\mathcal S}$ the unit of $\mathcal S$.
\medskip

A {\em precomposition} is a monoid morphism
$\circ:\bigbox\rightarrow \Hom(\mathcal S,\mathcal S)$ satisfying
\begin{eqnarray}
    \circ(\ShiftBoxnk{i,n}):\mathcal S_m\mathop\to
        \mathcal S_{n+m-1}, \label{precompeq5}\\
    \circ(\ShiftBoxnk{i,n})|_{\mathcal S_m}=
        Id_{\mathcal S_m} \mbox{ if }i \geq m + 1, \label{precompeq6}
\end{eqnarray}
where $\ |_{\mathcal S_m}$ denotes the restriction to $\mathcal S_m$.
For simplicity, we denote by $\Shiftnk{i,n}$ the map
$\circ(\ShiftBoxnk{i,n})$. Observe that the maps
$\Shiftnk{i,k}$ have the following intuitive meaning. If $s$ is an
element of $\mathcal{S}_m$, $\Shiftnk{i,n}(s)$ is an element of
$\mathcal{S}_{n + m - 1}$ obtained by inserting in $s$ a gap of length
$n - 1$ at position $i$. Axioms~\eqref{precompeq2}, \eqref{precompeq3},
and~\eqref{precompeq4} can be understood in the light of this
interpretation.

\begin{example}\label{Ex1}\rm
We consider the set $\mathcal V$ of infinite vectors with only a finite numbers of non zero entries. This set, endowed with the sums is a monoid. If $\mathcal V_n$ denotes the set of the vectors $v$ such that $m> n$ implies $v_m=0$, each $\mathcal V_n$ is a submonoid of $\mathcal V$. We define a precomposition  by setting $\Shiftnk{i,n}v=\mathfrak n_{i,n}v$.
For instance,
\[
\mathfrak n_{2,3}={\color{WildStrawberry}E_{1,1}}+{\color{PineGreen}E_{2,2}}+{\color{NavyBlue}\sum_{j\geq 3}E_{j+2,j}}=\left[
\begin{array}{cccccccc}
\color{WildStrawberry}1&0&0&0&0&\cdots&0&\cdots\\
0&\color{PineGreen}1&0&0&0&\cdots&0&\cdots\\
0&0&0&0&0&\cdots&0&\cdots\\
0&0&0&0&0&\cdots&0&\cdots\\
0&0&\color{NavyBlue}1&0&0&\cdots&0&\cdots\\
0  &0  &0  & \color{NavyBlue} 1  &0  &  \cdots  &0  &  \cdots\\
   &\vdots& & &\color{NavyBlue}\ddots&&\vdots &
\end{array}
 \right].
\]
Hence,
\[
\mathfrak n_{2,3}\left[\begin{array}{c}v_1\\v_2\\v_3\\v_4\\v_5\\v_6\\\vdots\end{array}\right]
=\left[\begin{array}{c}v_1\\v_2\\0\\0\\v_3\\v_4\\\vdots\end{array}\right]
\]
Notice that $\mathfrak n_{i,n}v$ for $i\leq 0$ is obtained by shifting down the entries of $v$ of $n-1
$ cells and by replacing the first $n-1$ entries by zero. For instance:
\[
\mathfrak n_{0,3}\left[\begin{array}{c}v_1\\v_2\\v_3\\v_4\\v_5\\v_6\\\vdots\end{array}\right]
=\left[\begin{array}{c}0\\0\\v_1\\v_2\\v_3\\v_4\\\vdots\end{array}\right]
\]
\end{example}
\medskip

\medskip

Let $\circ:\bigbox\rightarrow \Hom(\mathcal S ,\mathcal S)$ and
$\triangleright:\bigbox\rightarrow \Hom(\mathcal S' ,\mathcal S')$ be
two  precompositions. A map $\phi:\mathcal S\rightarrow\mathcal S'$ is a
{\em precomposition morphism} from $\circ$ to $\triangleright$ if
it is a monoid morphism and satisfies
\begin{eqnarray}
    \phi:\mathcal S_n\rightarrow \mathcal S'_n, \label{morph1}\\
    \ShiftTnk{i,n}(\phi(x))=\phi(\Shiftnk{i,n}(x))\label{morph2}.
\end{eqnarray}
We denote by $\Hom(\circ,\triangleright)$ the set of precomposition
morphisms from $\circ$ to $\triangleright$.
\medskip

It is easy to check that the class $\mathrm{PreComp}$ of precompositions endowed with the
    arrows $\Hom(\circ,\triangleright)$ for each
    $\circ, \triangleright\in\mathrm{PreComp}$ is a category.

Let $\circ$ be a precomposition, We define on $\mathcal S$ the binary operators $\circ^{(n)}_k$ by
\[
s\circ^{(n)}_i s'=\Shiftnk{i,n}s\oplus\Shiftnk{0,i}s'.
\]
We recall that $\oplus$ denotes the binary operation of the monoid $\mathcal S$.
\begin{lemma}
We have
\begin{itemize}
\item  For each  $n,m\geq 1$, $1\leq i<j\leq n$, $s,s''\in \mathcal S$ and $s'\in\mathcal S_m$, we have
\begin{equation}\label{ass1prec}
(s\circ^{(m)}_i s')\circ_{j+m-1}^{(n)}s''=(s\circ_j^{(n)} s'')\circ_i^{(m)}s'.
\end{equation}
\item For each $0< j\leq m$, $0<i\leq n$, and $s,s',s''\in\mathcal S$, we have
\begin{equation}\label{ass2prec}
(s\circ^{(m)}_i s')\circ^{(n)}_{i+j-1}s''=s\circ_{i}^{(n+m-1)}(s'\circ^{(n)}_js'').
\end{equation}
\end{itemize}
\end{lemma}
\begin{proof}
We have
\[
(s\circ_i^{(m)} s')\circ_{j+m-1}^{(n)}s''=
\Shiftnk{j+m-1,n}\Shiftnk{i,m}s\oplus \Shiftnk{j+m-1,n}\Shiftnk{0,i}s'\oplus\Shiftnk{0,j+m-1}s''.
\]
But (\ref{precompeq3}) implies $\Shiftnk{j+m-1,n}\Shiftnk{i,m}=\Shiftnk{i,m}\Shiftnk{j,n}$ and $\Shiftnk{j+m-1,n}\Shiftnk{0,i}=\Shiftnk{0,i}\Shiftnk{j-i+m,n}$. Since $i<j$ and $s'\in\mathcal S_m$, (\ref{morph2}) implies $\Shiftnk{j-i+m,n}s'=s'$. Furthemore by (\ref{precompeq4}), one has
$\Shiftnk{0,j+m-1}s''=\Shiftnk{i,m}\Shiftnk{0,j}s''$. Hence, we deduce
\[
(s\circ_i^{(m)} s')\circ_{j+m-1}^{(n)}s''=\Shiftnk{i,m}\Shiftnk{j,n} s\oplus \Shiftnk{i,m}\Shiftnk{0,j}s''\oplus \Shiftnk{0,i}s'=(s\circ_j^{(n)} s'')\circ_i^{(m)}s'.
\]
This proves  (\ref{ass1prec}).

Now let us prove (\ref{ass2prec}). We have
\[
(s\circ^{(m)}_i s')\circ^{(n)}_{i+j-1}s''=\Shiftnk{i+j-1,n}\Shiftnk{i,m}s\oplus
\Shiftnk{i+j-1,n}\Shiftnk{0,i}s'\oplus \Shiftnk{0,i+j-1}s''.
\]
But from (\ref{precompeq4}), one has $\Shiftnk{i+j-1,n}\Shiftnk{i,m}=\Shiftnk{i,m+n-1}$,
(\ref{precompeq3}) implies  $\Shiftnk{i+j-1,n}\Shiftnk{0,i}=\Shiftnk{0,i}\Shiftnk{j,n}$ and $\Shiftnk{0,i+j-1}=\Shiftnk{0,i}\Shiftnk{0,j}$. Hence,
\[
(s\circ^{(m)}_i s')\circ^{(n)}_{i+j-1}s''=\Shiftnk{i,m+n-1}s\oplus\Shiftnk{0,i}\Shiftnk{j,n}s'\oplus
\Shiftnk{0,i}\Shiftnk{0,j}s''.
\]
This proves  (\ref{ass2prec}).
\end{proof}

Observe also that we have
\begin{equation}\label{ass3prec}
\mathbf 0_{\mathcal S}\circ_1^{(n)}s=s\circ_i^{(n)} \mathbf 0_{\mathcal S}=s
\end{equation}
for each $i\geq 0,\ n\geq 1$, and $s\in\mathcal S$.

\medskip

\subsection{From precompositions to operads}
Let us consider a precomposition
$\circ:\bigbox\rightarrow \Hom(\mathcal S,\mathcal S)$. 
From the commutative monoid
$\mathcal S$, we define
\begin{equation}
    \mathbb S_n:=\left\{a_s^{(n)}:s\in \mathcal S_n\right\},
\end{equation}
and $\mathbb S:=\bigsqcup_{n \geq 1}\mathbb S_n$.
\begin{claim}
Each $\mathbb S_n$ is naturally endowed with a structure of monoid. Furthermore, we have
\[
\mathbb S_1\rightarrowtail \mathbb S_{2}\rightarrowtail \cdots \rightarrowtail\mathbb S_n\rightarrowtail\cdots
\]
where the arrows $\rightarrowtail$ denote injective morphisms of monoids.\\
Hence, $\mathcal S$ is nothing else but the inductive limit of the $\mathbb S_n$
\[
\mathcal S=\lim_{\longrightarrow}\mathbb S_n.
\]
\end{claim}
Now for any $1\leq i\leq n$, we define the binary operator
$\circ_i:\mathbb S_n\times \mathbb S_m\rightarrow \mathbb S_{k+m-1}$
by
\begin{equation}
    a_s^{(n)}\circ_ia_{s'}^{(m)}:=a_{s\circ_{i}^{(m)}s'}^{(n+m-1)}.
\end{equation}

\begin{proposition} \label{prop:construction_precomposition}
    The set $\mathbb S$ endowed with the partial compositions $\circ_i$
    is an operad.
\end{proposition}
\begin{proof}
From (\ref{ass3prec}), the unit of the structure is $\mathbf 1_{\mathbb S}:=a^{(1)}_{\mathbf 0_{\mathcal S}}$.  The associativity properties follows from (\ref{ass1prec}) and (\ref{ass2prec}).
\end{proof}
\medskip
\begin{example}
\rm Let us consider again the precomposition of Example \ref{Ex1}. The binary operators $\circ^{(m)}_i$ are described in terms of infinite matrices by
\[
v\circ_i^{(m)}v'=\mathfrak n_{i,n}v + \mathfrak n_{0,i}v'=
\left[\begin{array}{c}v_1\\\vdots\\ v_{i-1}\\ v_i+v'_1\\v'_2\\\vdots\\v'_{m}
\\v_{i+1}+v'_{m+1}\\v_{i+2}+v'_{m+2}\\\vdots\end{array}\right].
\]
Remark that each symbol $a_v^{(m)}$ can be identified with a vector $w$ of size $m$ such that $v_i=w_i$ for each $i\leq m$. Hence, the compositions are given by
\[
\left[
\begin{array}{c}v_1\\\vdots\\v_n \end{array}
\right]\circ_i\left[
\begin{array}{c}v'_1\\\vdots\\v'_m \end{array}
\right]=\left[
\begin{array}{c}v_1\\\vdots\\v_{i-1}\\v_i+v'_1\\v'_2\\\vdots\\v'_{m}\\v_{i+1}\\\vdots\\v_n \end{array}
\right].
\]
For instance, we can illustrate (\ref{Associativity1}) by remarking
\[
({\color{NavyBlue}v}\circ_i {\color{WildStrawberry}v'})\circ_{j+m-1}{\color{PineGreen} v''}=\left[
\begin{array}{c}\color{NavyBlue}v_1\\\vdots\\\color{NavyBlue}v_{i-1}\\{\color{NavyBlue}v_i}+\color{WildStrawberry}v'_1\\\color{WildStrawberry}v'_2\\\vdots\\\color{WildStrawberry}v'_{m}\\\color{NavyBlue}v_{i+1}\\\vdots\\\color{NavyBlue}v_{j-1}\\{\color{NavyBlue}v_{j}}+\color{PineGreen}v''_1\\
\color{PineGreen}v''_2\\\vdots\\\color{PineGreen}v''_p\\\color{NavyBlue}v_{j+1}\\\vdots\\ \color{NavyBlue}v_n\end{array}
\right]={\color{NavyBlue}v}\circ_i ({\color{WildStrawberry}v'}\circ_j {\color{PineGreen}v''}).
\]
for $i< j$, $v\in\B^n$, $v'\in\B^m$, and $v''\in\B^p$.
\end{example}
\medskip

We denote by $\OP(\circ)$ the operad $(\mathbb S,\circ_i)$ as defined in
the construction. For any $\phi\in \Hom(\circ,\triangleright)$, we define
\begin{equation}
    \phi^{\OP}:\OP(\circ)\rightarrow \OP(\triangleright)
\end{equation}
by
\begin{equation}
    \phi^{\OP}(a_s^{(n)})=a^{(n)}_{\phi(s)}.
\end{equation}
We check  that the following result holds:
\begin{claim}
    The arrow $\OP:\mathrm{PreComp}\rightarrow \mathrm{Operad}$ which
    associates with each precomposition $\circ$ the operad $\OP(\circ)$ and
    with each homomorphism $\phi\in \Hom(\circ,\triangleright)$ the operad
    morphism  $\phi^{\OP}$ is a functor.
\end{claim}
\medskip

\subsection{Quotients of precompositions\label{sec quot precomp}}
Let $\circ:\bigbox\rightarrow \Hom(\mathcal S,\mathcal S)$ be a
precomposition and $\gamma:\mathcal S\rightarrow\mathcal S$ be an
idempotent (that is, $\gamma^2=\gamma$) monoid morphism sending
$\mathcal S_n$ to $\mathcal S_n$ and satisfying:
$\Shiftnk{i,n}\gamma=\gamma\Shiftnk{i,n}$.
We define $\gamma:\mathbb S\rightarrow\mathbb S$ by
$\gamma a^{(n)}_s :=a^{(n)}_{\gamma s}$.
\medskip

\begin{lemma}\label{propgamma1}
The two following assertions hold:
\begin{enumerate}
\item for each $s\in\mathcal S_n$, $s'\in\mathcal S_m$, and $1\leq i\leq n$,
\begin{equation}
    \gamma\left(\gamma(a^{(n)}_s)\circ_i \gamma(a^{(m)}_{s'})\right)=
    \gamma\left(a^{(n)}_s\circ_i a^{(m)}_{s'}\right),
\end{equation}
\item $\gamma(s_1)=\gamma(s'_1)$ and $\gamma(s_2)=\gamma(s'_2)$ implies
\begin{equation}
    \gamma(a_{s_1}^{(n)}\circ_i a_{s_2}^{(m)})
    =\gamma(a_{s'_1}^{(n)}\circ_i a_{s'_2}^{(m)}).
\end{equation}
\end{enumerate}
\end{lemma}

\begin{proof}
\begin{enumerate}
\item We have $\gamma(a_s^{(n)})\circ_i \gamma(a_{s'}^{(m)})=
a_{\gamma s}^{(n)}\circ_i a_{\gamma s'}^{(m)}=a_{s''}^{(n+m-1)}$ with
$$s''=
\gamma(s)\circ_i^{(m)}\gamma(s')
= \Shiftnk{i,m}({\gamma s})\oplus \Shiftnk{0,i}({\gamma s'})=
\gamma\left(\Shiftnk{i,m}(s)\oplus \Shiftnk{0,i}(s')\right).
$$
Hence
$
\gamma\left(\gamma(a^{(n)}_s)\circ_i \gamma(a^{(m)}_{s'})\right)=a^{(n+m-1)}_{s'''}$ with
\[\begin{array}{rcl}s'''&=&
\gamma\left(\gamma\left(\Shiftnk{i,m}({s})\oplus \Shiftnk{0,i}({s'})\right)\right)
=\gamma\left(\Shiftnk{i,m}({s})\oplus \Shift{0,i}{s'}\right)\\
&=&\gamma\left(\Shift{i,m}{a_{s}^{(k)}}\oplus \Shift{0,i}{a_{s'}^{(m)}}\right)
=s''.
\end{array}
\]
\item Suppose $\gamma(s_1)=\gamma(s')$ and $\gamma(s_2)=\gamma(s'_2)$  then we have
 $$\gamma(a_{s_1}^{(n)}\circ_ia_{s_2}^{(m)})=\gamma\gamma(a_{s_1}^{(n)}\circ_ia_{s_2}^{(m)})= \gamma(a_{\gamma s_1}^{(n)}\circ_ia_{\gamma s_2}^{(m)})=\gamma(a_{\gamma s'_1}^{(n)}\circ_ia_{\gamma s'_2}^{(m)})=\gamma(a_{s'_1}^{(n)}\circ_ia_{s'_2}^{(m)}).
$$
\end{enumerate}
\end{proof}
\medskip

Consider now the equivalence relation $\sim_\gamma$ on $\mathcal{S}$
defined for any $s, s' \in \mathcal{S}$ by $s \sim_\gamma s'$ if and
only if $\gamma(s) = \gamma(s')$. By definition of $\gamma$, $\sim_\gamma$
is a monoid congruence of $\mathcal{S}$ and hence, $\mathcal{S}/_{{\sim_\gamma}}$
is a monoid. Consider also the equivalence relation $\equiv_\gamma$ on
$\OP(\circ)$ defined for any $a_s^{(n)}, a_{s'}^{(n)} \in \OP(\circ)$ by
$a_s^{(n)} \equiv_\gamma a_{s'}^{(n)}$ if and only if $s \sim_\gamma s'$.
Lemma \ref{propgamma1} shows that $\equiv_\gamma$ is actually an
operadic congruence and hence, that $\OP(\circ)/_{\equiv_\gamma}$ is an
operad.
\medskip

Let the precomposition
\begin{equation}
    \odot : \bigbox \rightarrow
    \Hom\left(\mathcal S/_{\sim_\gamma}, \mathcal S/_{\sim_\gamma}\right)
\end{equation}
defined for any $\sim_\gamma$-equivalence class $[s]_{\sim_\gamma}$
by $\Shiftdnk{i,n}\left([s]_{\sim_\gamma}\right):=[\Shiftnk{i,n}(s)]_{\sim_\gamma}$.
We then have
\begin{theorem}
The operads $\OP(\circ)/_{\equiv_\gamma}$ and $\OP(\odot)$ are isomorphic.
\end{theorem}
\begin{proof}
    Let us denote by $\circ^\gamma_i$ the composition map of
    $\OP(\circ)/_{\equiv_\gamma}$. Let the map
    \begin{equation}
        \phi : \OP(\circ)/_{\equiv_\gamma} \to \OP(\odot)
    \end{equation}
    defined for any $\equiv_\gamma$-equivalence class $[a_s^{(n)}]_{\equiv_\gamma}$
    by
    \begin{equation}
        \phi([a_s^{(n)}]_{\equiv_\gamma}) := a_{[s]_{\sim_\gamma}}^{(n)}.
    \end{equation}
    Let us show that $\phi$ is an operad morphism. For that, let
    $[a_s^{(n)}]_{\equiv_\gamma}$ and $[a_{s'}^{(m)}]_{\equiv_\gamma}$
    be two $\equiv_\gamma$ equivalence classes. One has
    \begin{equation} \label{equ:morphisme_gamma_1}
        \phi([a_s^{(n)}]_{\equiv_\gamma} \circ^\gamma_i [a_{s'}^{(m)}]_{\equiv_\gamma})
            \enspace = \enspace
                \phi([a_s^{(n)} \circ_i a_{s'}^{(m)}]_{\equiv_\gamma})
            \enspace = \enspace
                \phi([a_{s''}^{(n + m - 1)}]_{\equiv_\gamma})
            \enspace = \enspace a_{[s'']_{\sim_\gamma}}^{(n + m - 1)},    \end{equation}
    where $s'' := \Shiftnk{i,m}(s)\oplus \Shiftnk{0,i}(s')$.
    We moreover have
    \begin{equation} \label{equ:morphisme_gamma_2}
        \phi([a_s^{(n)}]_{\equiv_\gamma}) \odot_i
                \phi([a_{s'}^{(m)}]_{\equiv_\gamma})
        \enspace = \enspace
            a_{[s]_{\sim_\gamma}}^{(n)} \odot_i a_{[s']_{\sim_\gamma}}^{(m)}
        \enspace = \enspace a_{[s''']_{\sim_\gamma}}^{(n + m - 1)},
    \end{equation}
    where $[s''']_{\sim_\gamma} :=
        \Shiftdnk{i,m}([s]_{\sim_\gamma}) \oplus
        \Shiftdnk{0,i}([s']_{\sim_\gamma})$.
    Now, by using the fact that $\sim_\gamma$ is a monoid congruence,
    one has
    \begin{equation}
    \begin{split}
        [s''']_{\sim_\gamma}
        & = \Shiftdnk{i,m}([s]_{\sim_\gamma})
            \oplus \Shiftdnk{0,i}([s']_{\sim_\gamma}) \\
        & = [\Shiftnk{i,m}(s)]_{\sim_\gamma}
            \oplus [\Shiftnk{0,i}(s')]_{\sim_\gamma} \\
        & = [\Shiftnk{i,m}(s) \oplus \Shiftnk{0,i}(s')]_{\sim_\gamma} \\
        & = [s'']_{\sim_\gamma}.
    \end{split}
    \end{equation}
    This shows that \eqref{equ:morphisme_gamma_1} and
    \eqref{equ:morphisme_gamma_2} are equal and hence, that $\phi$ is
    an operad morphism.
    \smallskip

    Furthermore, the definitions of $\sim_\gamma$ and $\equiv_\gamma$
    imply that $\phi$ is a bijection. Therefore, $\phi$ is an operad
    isomorphism.
\end{proof}
\medskip

\section{Multi-tildes and precompositions}\label{sec4}
In \cite{LMN12}, we investigated several operads allowing to describe the behaviour of the multi-tilde operators. In this section, we show that some of them admit an alternative definition using the notion of precomposition.
\medskip

\subsection{The operad $\mathcal{T}$ revisited\label{Trev}}
For any $m \geq 1$, let $\mathcal{S}^{\mathcal{T}}_m$ be the set of
subsets of ${\{(x,y): 1\leq x\leq y\leq m\}}\subset \N^2$. Noting that
$\mathcal{S}^{\mathcal{T}}_m\subset \mathcal{S}^{\mathcal{T}}_{m+1}$ we
define
$\mathcal{S}^{\mathcal{T}}:=\bigcup_{m\geq 1} \mathcal{S}^{\mathcal{T}}_m$.
Considering the binary operation $\cup$ as a product, the pair
$(\mathcal{S}^{\mathcal{T}},\cup)$ defines a commutative monoid whose
unit is $\mathbf{1}_{\mathcal{S}^{\mathcal{T}}}=
\emptyset\in \mathcal{S}^{\mathcal{T}}_1$. This  monoid is
generated by the set $\{\{(x,y)\}:{1\leq x\leq y}\}$.
\medskip

Now define $\circ:\ \bigbox\ \rightarrow\
\Hom(\mathcal{S}^{\mathcal{T}},\mathcal{S}^{\mathcal{T}})$ by
$
    \circ(\ShiftBoxnk{i,n}):=\Shiftnk{i,n}
$
where each homomorphism $\Shiftnk{i,n}$ is defined by its values on the
generators
\begin{equation}
    \Shiftnk{i,n}(\{(x,y)\}):=
    \begin{cases}
        \{(x,y)\} & \text{ if } y<i,\\
        \{(x,y+n-1)\} & \text{ if } x\leq i\leq y,\\
        \{(x+n-1,y+n-1)\} & \text{ otherwise.}\\
      \end{cases}
\end{equation}

Remark that $\circ$ is a monoid morphism. Indeed,
\begin{enumerate}
\item The set of the homomorphisms $\Shiftnk{i,n}$ generates a submonoid of
$\Hom(\mathcal{S}^{\mathcal{T}},\mathcal{S}^{\mathcal{T}})$ (with  $\mathrm{Id}_{\mathcal{S}^{\mathcal{T}}}$ as unit)
\item By construction,
      $\Shiftnk{i,n}:\ \mathcal{S}^{\mathcal{T}}_m\ \rightarrow\ \mathcal{S}^{\mathcal{T}}_{m+n-1}$ and
      $\Shiftnk{i,n}|_{\mathcal{S}^{\mathcal{T}}_m} =\mathrm{Id}_{\mathcal{S}^{\mathcal{T}}_m}$ if $m<i$.
   \item
The operators $\Shiftnk{i,n}$ satisfy (see \cite{LMN12})    \begin{itemize}
      \item $\Shiftnk{i,n}=\Shiftnk{0,n}$ for each $i<0$,
      \item $\Shiftnk{i,1}=\Shiftnk{0,1}=\mathrm{Id}_{\mathcal{S}^{\mathcal{T}}}$ for each $i$,
      \item $\Shiftnk{i,n}\Shiftnk{j,n'}=\Shiftnk{j+n-1,n'}\Shiftnk{i,n}$ if $i\leq j$ or $i,j\leq 0$,
      \item $\Shiftnk{i+j,n}\Shiftnk{i,n'}=\Shiftnk{i,n+n'-1}$ if $0\leq j<n'$.
    \end{itemize}
    \end{enumerate}
Hence $\circ$ is a precomposition.
More precisely, the operad $\mathcal T$ can be seen as the operad constructed from the precomposition $\circ$:
\begin{proposition}
  The operads $\mathcal{T}$ and  $\mathrm{\OP}(\circ)$ are isomorphic.
\end{proposition}
\begin{proof}
  The isomorphism is given by the map from $\mathcal{T}_m$ to $\mathbb{S}_m$ sending any element $T$ to $a^{(k)}_T$.
\end{proof}
\medskip

\subsection{The operad $\RAS$ revisited}
In \cite{LMN12}, we considered an operad $\RAS$ on reflexive and antisymmetric relations that are compatible with the natural order on integers (\emph{i.e.}, $(x,y)\in\RAS$ implies $x\leq y$). Since the elements $(x,x)$ do not play any role in the construction, we propose here an alternative construction based on antireflexive and antisymmetric relations.
\medskip

For any $m \geq 1$, let $\mathcal{S}^{\Diamond}_m$ be the set of
subsets of $\{(x,y): 1\leq x< y\leq m+1\}\subset\N^2$.
By construction we have  $\mathcal{S}^{\Diamond}_n\subset \mathcal{S}^{\Diamond}_{m+1}$.
Endowed with the binary operation $\cup$, the set $\mathcal{S}^{\Diamond}:=\bigcup_{m\geq 1} \mathcal{S}^{\Diamond}_m$ is a    commutative monoid generated by $\{\{(x,y)\}_{1\leq x< y}\}$.
\medskip

Let us define $\diamond:\ \bigbox\ \rightarrow\
\Hom(\mathcal{S}^{\Diamond},\mathcal{S}^{\Diamond})$ by
$\diamond(\ShiftBoxnk{i,n}):=\ShiftDiamnk{i,n}$
with
\begin{equation}
    \ShiftDiamnk{i,k}(\{(x,y)\}):=
    \begin{cases}
        \{(x,y)\} & \text{ if } y\leq i,\\
        \{(x,y+n-1)\} & \text{ if } x\leq i< y,\\
        \{(x+n-1,y+n-1)\} & \text{ otherwise.}\\
      \end{cases}
\end{equation}
Figure \ref{FigShiftSimple} illustrates the action of the operator $\ShiftDiamnk{i,n}$  as an operation cutting one triangle representing the set of all the  couples $(x,y)$ with $1\leq x<y\leq m+1$ into two triangles
and a rectangle and putting them back into a larger triangle.
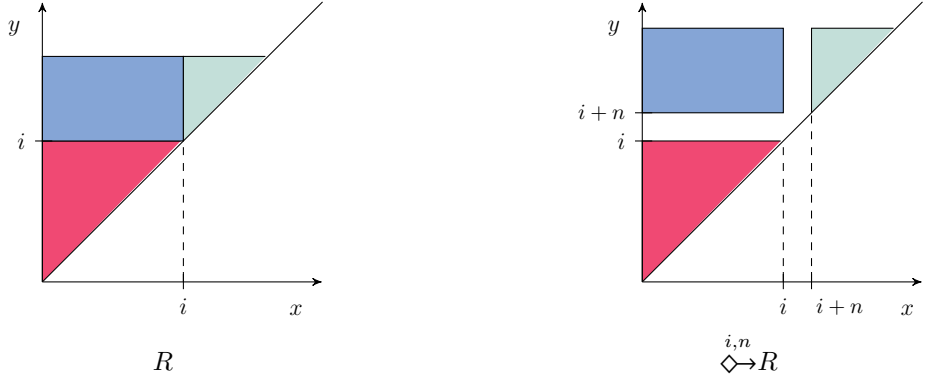
\begin{figure}[h]
\begin{center}
\begin{tabular}{c@{\ \ \ \ \ \ \ \ \ \ \ \ \ \ \ \ \ \ \ \ \ \ \ \ \ \ \ \ }c}

   \begin{tikzpicture}[scale=0.75,node distance=2.5cm,bend angle=30,transform shape,scale=1]
				\filldraw[fill=WildStrawberry!90] (0,0) --(0,2.5) -- (2.5,2.5);
				\filldraw[fill=NavyBlue!50] (0,2.5) rectangle (2.5,4);
				\filldraw[fill=PineGreen!20] (2.5,2.5) -- (2.5,4) -- (4,4);
				\draw[->] (0,0) -- (0,5);
				\draw[->] (0,0) -- (5,0);
				\draw (0,0) -- (5,5);
				\draw[dashed] (2.5,0) -- (2.5,2.5);
				\draw (2.5,5pt) -- (2.5,-5pt) node[below] {\Large $i$};
				\draw (5pt,2.5) -- (-5pt,2.5) node[left] {\Large $i$};
				\node at (-0.5,4.5) [name=y]{\Large $y$};
				\node at (4.5,-0.5) [name=x]{\Large $x$};
      \end{tikzpicture}
   &
   \begin{tikzpicture}[scale=0.75,node distance=2.5cm,bend angle=30,transform shape,scale=1]
				\filldraw[fill=WildStrawberry!90] (0,0) --(0,2.5) -- (2.5,2.5);
				\filldraw[fill=NavyBlue!50] (0,3) rectangle (2.5,4.5);
				\filldraw[fill=PineGreen!20] (3,3) -- (3,4.5) -- (4.5,4.5);
				\draw[->] (0,0) -- (0,5);
				\draw[->] (0,0) -- (5,0);
				\draw (0,0) -- (5,5);
				\draw[dashed] (2.5,0) -- (2.5,2.5);
				\draw[dashed] (3,0) -- (3,3);
				\draw (2.5,5pt) -- (2.5,-5pt) node[below] {\Large $i$};
				\draw (5pt,2.5) -- (-5pt,2.5) node[left] {\Large $i$};
				\draw (3,5pt) -- (3,-5pt) node[below,xshift=0.5cm] {\large $i+n$};
				\draw (5pt,3) -- (-5pt,3) node[left] {\large $i+n$};
				\node at (-0.5,4.5) [name=y]{\Large $y$};
				\node at (4.7,-0.5) [name=x]{\Large $x$};
      \end{tikzpicture}
   \\
$R$&$\ShiftDiamnk{i,n}R$
\end{tabular}
\caption{Action of $\RAS$ on pairs.\label{FigShiftSimple}}
\end{center}
\end{figure}

Similarly to Section \ref{Trev}, we consider the submonoid of
$\Hom(\mathcal{S}^{\Diamond},\mathcal{S}^{\Diamond})$ generated by the
elements $\ShiftDiamnk{i,n}$ . We have
$\ShiftDiamnk{i,n}:\ \mathcal{S}^{\Diamond}_m\ \rightarrow\
\mathcal{S}^{\Diamond}_{n+n-1}$ and
$\ShiftDiamnk{i,n}|_{\mathcal{S}^{\Diamond}_m}
=\mathrm{Id}_{\mathcal{S}^{\Diamond}_m}$ when $m<i$. Furthermore, the
elements $\ShiftDiamnk{i,n}$ satisfy the properties
\begin{itemize}
    \item $\ShiftDiamnk{i,n}=\ShiftDiamnk{0,n}$ for each $i<0$,
    \item
    $\ShiftDiamnk{i,1}=\ShiftDiamnk{0,1}=\mathrm{Id}_{\mathcal{S}^{\Diamond}}$
    for each $i$
    \item
    $\ShiftDiamnk{i,n}\ShiftDiamnk{j,n'}=\ShiftDiamnk{j+n-1,n'}\ShiftDiamnk{i,n}$
    if $i\leq j$ or $i,j\leq 0$
    \item $\ShiftDiamnk{i+j,n}\ShiftDiamnk{i,n'}=\ShiftDiamnk{i,n+n'-1}$
    if  $0\leq j<n'$.
\end{itemize}
The map $\diamond$ is a monoid morphism  and so a precomposition.
We set $\mathrm{ARAS}:=\mathrm{\OP}(\diamond)=(\mathbb{S}^\Diamond,\diamond)$. The operad $\mathrm{ARAS}$ is an alternative  construction for the operad $\RAS$ as shown by:

\begin{proposition}
    The operads $\RAS$ and  $\mathrm{ARAS}$ are isomorphic.
\end{proposition}
\begin{proof}
  The isomorphism is given by the map from $\RAS_n$ to
  $\mathbb{S}^\Diamond_m$ sending any element $R$ to
  $a^{(m)}_{R\setminus\Delta}$ where $\Delta:=\{(x,x):x\in\mathbb{N}\}$.
\end{proof}
\medskip

\subsection{The operad $\mathrm{POSet}$ revisited}
The operad $\mathrm{POSet}$ is defined as a quotient of the operad $\RAS$.
In~\cite{LMN12}, we showed that $\mathrm{POSet}$ is optimal in the sense
that two of its operators have two different actions on languages.
\medskip

Denote by $\gamma: S^{\Diamond} \rightarrow S^{\Diamond}$ the transitive
closure. Remarking that $\gamma(R): S^{\Diamond}_n \rightarrow S^{\Diamond}_n$
and $\ShiftDiamnk{i,n} \gamma=\gamma\ShiftDiamnk{i,n} $, we apply the
result of Section \ref{sec quot precomp} and define the precomposition
$\Diamonddot: \bigbox \to
    \Hom(S^{\Diamond}_{/\equiv_\gamma},
    S^{\Diamond}_{/\equiv_\gamma})$
by setting
$\ShiftDiamDnk{i,n}([R]) := \left[\ShiftDiamnk{i,n}(R)\right]$
where $[]:S^{\Diamond}\rightarrow S^{\Diamond}_{/\equiv_\gamma}$ denotes
the natural morphism sending each element $R$ of $S^{\Diamond}$ to its
equivalence class $[R]$.
\medskip

The operad $\mathrm{\OP}(\Diamonddot)$ gives an alternative way to define
the operad $\mathrm{POSet}$ using precompositions.
\begin{proposition}
    The operads $\mathrm{POSet}$, $\mathrm{\OP}(\Diamonddot)$ and
    $\mathrm{ARAS}_{/\equiv_\gamma}$ are isomorphic.
\end{proposition}
\begin{proof}
    The isomorphism is given by the map from $\mathrm{POSet}_n$ to
    $\mathbb{S}^\Diamonddot_n$ sending any element $P$ to
    $a^{(n)}_{\left[\left(P\setminus\Delta\right)\right]}$ where
    $\Delta=\{(x,x):x\in\mathbb{N}\}$.
\end{proof}
\medskip

\section{The operad of double multi-tildes}\label{sec5}
In \cite{LMN12}, we proved that the action of $\mathcal T$ on symbols
allows us to denote all finite languages. In this section, we propose an
extension of the operad $\mathcal T$ in order to represent  infinite
languages. New operators are required in order to describe the Kleene
star operation $^*$. In the last section of \cite{LMN12}, we  introduced
an operad $\mathcal T^*$ generated by $\mathcal T$ together with an
additional operator $\star$ (denoting the Kleene star $^*$).  Albeit this
operad allows the manipulation of regular languages, the equivalence of
the operators, w.r.t. the action over languages, is difficult to model.
In this section, we introduce a new operad $\mathcal {DT}$ which is
composed of two kinds of multi-tildes: right and left multi-tildes. The
$^*$ operation will be realized by a combination of right and left
multi-tildes operations.  Furthermore, we show that the expressiveness
of these operators is higher than operators of $\mathcal T^\star$ for a
given number of symbols. We start by considering that the two types of
operators are independently composed. More precisely,
\begin{equation} \label{DT}
    \mathcal{DT}:=\Had(\mathcal T,\mathcal T).
\end{equation}
We mimic the construction of \cite{LMN12} linking multi-tildes and
reflexive antisymmetric relations in order to construct a new operad
$\mathrm{ARef}$, which elements are antireflexive relations, isomorphic
to~$\mathcal{DT}$.
\medskip

\subsection{$\mathcal{DT}$ and antireflexive relations\label{SDT2ARef}}
We consider the graded set
\begin{equation}
\mathcal S^{\mathrm{ARef}}:=\bigcup_{n \geq 1}\mathcal S^{\mathrm{ARef}}_n,
\end{equation}
where $\mathcal{S}^{\ARef}_n$ is the set of subsets of
${\{(x,y): 1\leq x,y\leq n+1,\, x\neq y\}}\subset \N^2$.
\def\rev{\mathrm{rev}}
Endowed with the binary operation $\cup$, the set $\mathcal S^{\mathrm{ARef}}$ is a commutative monoid generated by $\{(x,y):x\neq y\}$.
We define the map $\Diamond: \bigbox\rightarrow  \Hom(\mathcal S^{\mathrm{ARef}},\mathcal S^{\mathrm{ARef}})$ by $\Diamond(\ShiftBoxnk{i,n})=\ShiftDiamnk{i,n}$
where
\begin{equation}
\ShiftDiamnk{i,n}(\{(x,y)\}) :=
    \begin{cases}
        \{(x,y)\} & \text{ if } x,y\leq i,\\
        \{(x,y+n-1)\} & \text{ if } x\leq i\mbox{ and }i< y,\\
        \{(x+n-1,y)\} & \text{ if } i<x\mbox{ and } y\leq i,\\
        \{(x+n-1,y+n-1\} & \text{ otherwise.}\\
      \end{cases}
\end{equation}
We easily check  that $\Diamond$ is a precomposition and we set $\ARef:=\mathrm{\OP}(\Diamond)$. Observe that a graphical representation of the  action of $\ARef$ can be obtained from those of $\ARAS$ by replacing triangles by squares (see Figure~\ref{FigShiftDouble}).
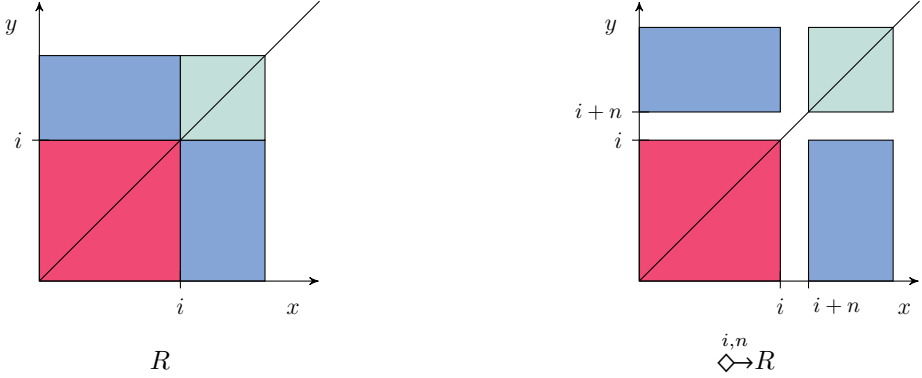
\begin{figure}[h]
\begin{center}
\begin{tabular}{c@{\ \ \ \ \ \ \ \ \ \ \ \ \ \ \ \ \ \ \ \ \ \ \ \ \ \ \ \ }c}
   \begin{tikzpicture}[scale=0.75,node distance=2.5cm,bend angle=30,transform shape,scale=1]
				\filldraw[fill=WildStrawberry!90] (0,0) rectangle (2.5,2.5);
				\filldraw[fill=NavyBlue!50] (0,2.5) rectangle (2.5,4);
				\filldraw[fill=NavyBlue!50] (2.5,0) rectangle (4,2.5);
				\filldraw[fill=PineGreen!20] (2.5,2.5) rectangle (4,4);
				\draw[->] (0,0) -- (0,5);
				\draw[->] (0,0) -- (5,0);
				\draw (0,0) -- (5,5);
				\draw (2.5,5pt) -- (2.5,-5pt) node[below] {\Large $i$};
				\draw (5pt,2.5) -- (-5pt,2.5) node[left] {\Large $i$};
				\node at (-0.5,4.5) [name=y]{\Large $y$};
				\node at (4.5,-0.5) [name=x]{\Large $x$};
      \end{tikzpicture}
   &   \begin{tikzpicture}[scale=0.75,node distance=2.5cm,bend angle=30,transform shape,scale=1]
				\filldraw[fill=WildStrawberry!90] (0,0) rectangle (2.5,2.5);
				\filldraw[fill=NavyBlue!50] (0,3) rectangle (2.5,4.5);
				\filldraw[fill=NavyBlue!50] (3,0) rectangle (4.5,2.5);
				\filldraw[fill=PineGreen!20] (3,3) rectangle (4.5,4.5);
				\draw[->] (0,0) -- (0,5);
				\draw[->] (0,0) -- (5,0);
				\draw (0,0) -- (5,5);
				\draw (2.5,5pt) -- (2.5,-5pt) node[below] {\Large $i$};
				\draw (5pt,2.5) -- (-5pt,2.5) node[left] {\Large $i$};
				\draw (3,5pt) -- (3,-5pt) node[below,xshift=0.5cm] {\large $i+n$};
				\draw (5pt,3) -- (-5pt,3) node[left] {\large $i+n$};
				\node at (-0.5,4.5) [name=y]{\Large $y$};
				\node at (4.7,-0.5) [name=x]{\Large $x$};
      \end{tikzpicture}
   \\
$R$&$\ShiftDiamnk{i,n}R$
\end{tabular}
\caption{Action of $\ARef$ on pairs.\label{FigShiftDouble}}
\end{center}
\end{figure}
\begin{proposition}\label{ARefARAS}
The operad $\ARef$ is isomorphic to  $\Had(\ARAS,\ARAS)$.
\end{proposition}
\begin{proof}
Let us denote $\rev(x,y)=(y,x)$. If $R\in \ARef$ we will denote $\rev(R)=\{\rev(x,y):(x,y)\in R\}$, $R^<=\{(x,y)\in R:x< y\}$, and $R^>=\{(x,y)\in R:x> y\}=\rev((\rev(R))^<)$ (note that $R=R^<\cup R^>$).
Let $\Phi:\Had(\ARAS,\ARAS)\rightarrow \ARef$ be the map defined by $\Phi(a^{(n)}_{R_1},a^{(n)}_{R_2})=
a^{(n)}_{R_1\cup \rev(R_2)}$.
This map is a bijection which inverse is $\Phi^{-1}(a_R^{(n)})=(a^{(n)}_{R^<},a^{(n)}_{\rev(R^>)})$.\\
Let us prove that $a_R^{(n)}\Diamond_i a_{R'}^{(m)}=\Phi(\Phi^{-1}(a_R^{(n)})\Diamond_i\Phi^{-1}(a_{R'}^{(m)}))$.
We have
\[
\Phi(\Phi^{-1}(a_R^{(n)})\Diamond_i\Phi^{-1}(a_{R'}^{(m)}))=
\Phi(\Phi^{-1}(a_{R''}^{(n+m-1)}))=
\Phi(a_{R''^<}^{(n+m-1)}, a_{\rev(R''^>)}^{(n+m-1)}),
\]
where $R''= \ShiftDiam{i,m}{R}\cup \ShiftDiam{0,i}{R'}$.
Let $\ast\in\{<,>\}$. Since, $\ShiftDiam{i,k'}{R^{\ast}}=(\ShiftDiam{i,m}R)^{\ast}$ and $\ShiftDiam{0,i}{R'^{\ast}}=(\ShiftDiam{0,i}{R'})^{\ast}$, we have  $(\ShiftDiam{i,m}{R^{\ast}}\cup \ShiftDiam{0,i}{R'^{\ast}})=(R'')^{\ast}$. In other words, $R''^{\ast}=\ShiftDiam{i,m}{R^{\ast}}\cup \ShiftDiam{0,i}{R'^{\ast}}$. Hence,
\[
\Phi(\Phi^{-1}(R)\Diamond_i\Phi^{-1}(R'))=\Phi(a^{(n+m-1)}_{R''^{<}},a^{(n+m-1}_{\rev(R''^{>})})
=a^{(n+m-1}_{R''^<\cup R''^>}=a^{(n+m-1)}_{R''}=a_R^{(n)}\Diamond_i a_{R'}^{(m)}.
\]
This proves that $\ARef$ is an operad isomorphic to $\Had(\ARAS,\ARAS)$.
\end{proof}

\begin{corollary}
The operads $\mathcal{DT}$, $\ARef$, $\Had(\RAS,\RAS)$, and $\Had(\ARAS,\ARAS)$  are isomorphic.
\end{corollary}

In the aim to illustrate the isomorphism between $\ARef$ and $\mathcal{DT}$, we recall that the graded map $\zeta:\mathcal T_n\rightarrow \RAS_n$ defined by $\zeta(R)=\{(x,y+1):(x,y)\in R\}\cup \{(1,1),\dots,(n+1,n+1)\}$ is an isomorphism of operads. According to the definition of $\ARAS$, we obtain explicitly an isomorphism from $\mathcal T$ to $\ARAS$ by a slight modification of $\zeta$:
$
\zeta^A(R)=a^{(n)}_{\zeta(R)\setminus\Delta}.
$
Since $\ARAS$ and $\mathcal T$ are isomorphic, this is also the case for $\mathcal{DT}$ and $\ARef$ (because $\ARef$ is isomorphic to $\Had(\ARAS,\ARAS)$). From the construction described in Proposition \ref{ARefARAS}, the map $\xi: \mathcal{DT}\rightarrow \ARef$ defined by $\xi(R_1,R_2)=a^{(n)}_{\zeta^A(R_1)\cup \rev(\zeta^A(R_2))}$, when $(R_1,R_2)\in\mathcal{DT}_n$, explicits the isomorphism.
\medskip

\begin{example}\rm\label{ExARAS}
Consider $P_1=(\{(1,3),(2,2),(3,4)\},\{(2,3)\})\in \mathcal{DT}_5$ and $P_2=(\{(2,3),(3,4)\},\{(1,2),(3,4)\})\in\mathcal{DT}_4$. We have
\[
\xi(P_1)=a^{(5)}_{\{(1,4),(2,3),(3,5),(4,2)\}}\mbox{ and }\xi(P_2)=a^{(4)}_{\{(2,4),(3,5),(3,1),(5,3)\}}.
\]
Remark that
$$\begin{array}{rcl}P_1\circ_2 P_2&=&(\{(1,3),(2,2),(3,4)\}\circ_2\{(2,3),(3,4)\},\{(2,3)\}\circ_2\{(1,2),(3,4)\})\\
&=& (\{(1,6),(2,5),(6,7),(3,4),(4,5)\},\{(2,6),(2,3),(4,5)\}),
\end{array}$$
and then
\[
\xi(P_1\circ_2 P_2)=a^{(8)}_{\{(1,7),(2,6),(6,8),(3,5),(4,6),(7,2),(4,2),(6,4)\}}.
\]
Let us now compute $\xi(P_1)\Diamond_2\xi(P_2)$:
\[
\xi(P_1)\Diamond_2\xi(P_2)=a^{(5)}_{\{(1,4),(2,3),(3,5),(4,2)\}}\Diamond_2a^{(4)}_{\{(2,4),(3,5),(3,1),(5,3)\}}=a^{(8)}_{R}
\]
with
\[\begin{array}{rcl}
R&=&\ShiftDiam{2,4}{\{(1,4),(2,3),(3,5),(4,2)\}}\cup\ShiftDiam{0,2}{\{(2,4),(3,5),(3,1),(5,3)\}}\\
&=& \{(1,7),(2,6),(6,8),(7,2),(3,5),(4,6),(4,2),(6,4)\}.
\end{array}
\]
We observe that $\xi(P_1\circ_2 P_2)=\xi(P_1)\Diamond_2\xi(P_2)$.\\
Graphically, the composition $\Diamond_i$ can be illustrated in two steps corresponding to the operators $\ShiftDiamnk{ik'}$ and $\ShiftDiamnk{0,i}$ by drawing the graph of the relations. For instance, we start with the two graphs of the relations $\{(1,4),(2,3),(3,5),(4,2)\}$ and $\{(2,4),(3,5),(3,1),(5,3)\}$:
\begin{figure}[H]
  \hfill
  \begin{minipage}{0.33\linewidth}
    \centerline{
      \begin{tikzpicture}[node distance=1.25cm,bend angle=30]
            \node (1) {$1$};
            \node[above of=1] (2) {$2$};
            \node[above right of=2] (3) {$3$};
            \node[right of=3] (4) {$4$};
            \node[below right of=4] (5) {$5$};
            \node[below of=5] (6) {$6$};
            \path[->]
              (1) edge (4)
              (2) edge (3)
              (3) edge (5)
              (4) edge (2)
            ;
      \end{tikzpicture}
    }
  \end{minipage}
  \begin{minipage}{0.33\linewidth}
    \centerline{
      \begin{tikzpicture}[node distance=1.25cm,bend angle=30,color=WildStrawberry]
            \node (1) {$1$};
            \node[above of=1] (2) {$2$};
            \node[above right of=2] (3) {$3$};
            \node[right of=3] (4) {$4$};
            \node[below right of=4] (5) {$5$};
            \path[->]
              (2) edge (4)
              (3) edge (1)
              (3) edge (5)
              (5) edge (3)
            ;
      \end{tikzpicture}
    }
  \end{minipage}
  \hfill
  \hfill
\end{figure}

We rename the vertices $3\rightarrow 6$, $4\rightarrow 7$, $\dots$, $6\rightarrow 9$ in the graphs of
$\{(1,4),(2,3),(3,5),(4,2)\}$ and the vertices $1\rightarrow 2, \dots,5\rightarrow 6$ in the graph of $\{(2,4),(3,5),(3,1),(5,3)\}$.
\begin{figure}[H]
  \hfill
  \begin{minipage}{0.33\linewidth}
    \centerline{
      \begin{tikzpicture}[node distance=1.25cm,bend angle=30]
            \node (1) {$1$};
            \node[above of=1] (2) {$2$};
            \node[above right of=2] (3) {$6$};
            \node[right of=3] (4) {$7$};
            \node[below right of=4] (5) {$8$};
            \node[below of=5] (6) {$9$};
            \path[->]
              (1) edge (4)
              (2) edge (3)
              (3) edge (5)
              (4) edge (2)
            ;
      \end{tikzpicture}
    }
  \end{minipage}
  \begin{minipage}{0.33\linewidth}
    \centerline{
      \begin{tikzpicture}[node distance=1.25cm,bend angle=30,color=WildStrawberry]
            \node (1) {$2$};
            \node[above of=1] (2) {$3$};
            \node[above right of=2] (3) {$4$};
            \node[right of=3] (4) {$5$};
            \node[below right of=4] (5) {$6$};
            \path[->]
              (2) edge (4)
              (3) edge (1)
              (3) edge (5)
              (5) edge (3)
            ;
      \end{tikzpicture}
    }
  \end{minipage}
  \hfill
  \hfill
\end{figure}

Then we identify the  vertices which have the same label in the two graphs:
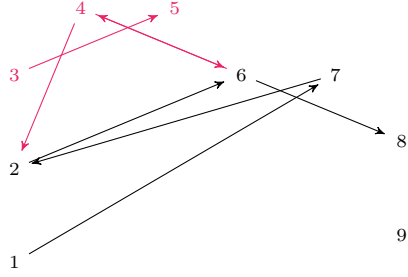
\begin{figure}[H]
    \centerline{
      \begin{tikzpicture}[node distance=1.25cm,bend angle=30]
            \node (1) {$1$};
            \node[above of=1] (2) {$2$} ;
            \node[above of=2, color=WildStrawberry] (3) {$3$};
            \node[above right of=3, color=WildStrawberry] (4) {$4$};
            \node[right of=4, color=WildStrawberry] (5) {$5$};
            \node[below right of=5] (6) {$6$};
            \node[right of=6] (7) {$7$};
            \node[below right of=7] (8) {$8$};
            \node[below of=8] (9) {$9$};
            \path[->]
              (1) edge (7)
              (2) edge (6)
              (7) edge (2)
              (6) edge (8)
            ;
            \path[->, color=WildStrawberry]
              (3) edge (5)
              (4) edge (2)
              (4) edge (6)
              (6) edge (4)
            ;
      \end{tikzpicture}
    }
\end{figure}
\end{example}
\medskip

\subsection{An operad on quasiorders}
A {\em quasiorder} is a reflexive and transitive relation. If $R$ is a
relation we denote by $\gamma(R)$ its transitive closure. We  also set
$\gamma^{\mathrm{A}}(R)=\gamma(R)\setminus\{(n,n):n\geq 1\}$ and
$\gamma^{\mathrm{R}}(R)=\gamma(R)\cup \{(n,n):n\geq 1\}$. Note that
$\gamma^{\mathrm{R}}(R)$ is the smallest quasiorder which contains $R$.
Since $\gamma^{\mathrm A}:\mathcal S^{\ARef}\rightarrow \mathcal S^{\ARef}$
is an idempotent monoid morphism sending $S^{\ARef}_n$ to $S^{\ARef}_n$
and satisfies $\ShiftDiamnk{i,n}\gamma^A=\gamma^A\ShiftDiamnk{i,n}$,
following Section~\ref{sec quot precomp}, we construct the precomposition
\begin{equation}
    \Diamonddot:\bigbox\to
    \Hom\left(\mathcal S^{\ARef}/_{\equiv_{\gamma^A}},
        \mathcal S^{\ARef}/_{\equiv_{\gamma^A}}\right)
\end{equation}
defined by $\ShiftDiamD{i,n}{[R]}:=\left[\ShiftDiam{i,n}R\right]$ where
$[]$ denotes the natural morphism
$\mathcal S^{\ARef}\rightarrow \mathcal S^{\ARef}/_{\equiv_{\gamma^A}}$
sending each relation to its class. Hence, we consider the operad
$\mathrm{\OP}(\Diamonddot)$.
\medskip

Alternatively, consider the set $\QOSET_n$ of quasiorder of $\{1,\dots,n+1\}$ and $\QOSET:=\bigcup_n\QOSET_n$. Consider also the partial composition defined by $Q\Diamonddot_i Q'=\gamma(\ShiftDiam{i,m}Q\cup
\ShiftDiam{0,i}{Q'})$ if $Q\in \QOSET_{n}$, $Q'\in \QOSET_{m}$ and $i\leq n$.
\medskip

\begin{theorem}
 The pair $(\QOSET,\Diamonddot)$ is an operad isomorphic to $\mathrm{\OP}(\Diamonddot)$.
\end{theorem}
\begin{proof}
Consider the map $\eta: \QOSET\rightarrow  \mathrm{\OP}(\Diamonddot)$ given by $\eta(Q)=a^{(n)}_{\left[Q\setminus\Delta\right]}$. The map $\eta$ is a graded bijection and its inverse is given by $\eta^{-1}(a^{(n)}_{[R]})=\gamma^{\mathrm R}(R)$. Remarking that
\begin{equation}\begin{split}
\eta^{-1}(a^{(n)}_{\left[R'\right]} \Diamonddot_i a^{(n)}_{\left[R''\right]})
    & =\gamma^{\mathrm R}(R'\Diamonddot_iR'') \\
    & =\gamma(\ShiftDiam{i,m}{\gamma^{\mathrm R}(R')})\cup
\ShiftDiam{0,i}{\gamma^{\mathrm R}(R'')} \\
    & =\gamma^{\mathrm R}(R')\Diamond_i\gamma^{\mathrm R}(R'') \\
    & =
\eta^{-1}(a^{(n)}_{\left[R'\right]})\Diamond_i\eta^{-1}(a^{(n)}_{\left[R''\right]}),
\end{split}\end{equation}
we prove that the set $\QOSET$ inherits from $\mathrm{\OP}(\Diamonddot)$ of a structure of operad.
\end{proof}
\medskip

\begin{example}\rm
Let us give an example. Consider, as in Example \ref{ExARAS}, the antireflexive relations $R_1=\{(1, 4), (2, 3), (3, 5), (4, 2)\}$ and $R_2=\{(2,4),(3,5),(3,1),(5,3)\}$. We have
\[
\gamma^{\mathrm R}(R_1)=\{(1,4),(2,3),(3,5),(4,2),(1,2),(2,5),(4,3),(1,3),(4,5),(1,5),(1,1),(2,2),(3,3),(4,4),(5,5),(6,6)\}\] and
\[
\gamma^{\mathrm R}(R_2)=\{(2,4),(3,5), (3,1),(5,3),(5,1),(1,1),(2,2),(3,3),(4,4),(5,5)\}.
\]
We have
\[\begin{array}{rcl}
\gamma^{\mathrm R}(R)&=&\gamma^{\mathrm R}(\ShiftDiam{2,4}{R_1}\cup\ShiftDiam{0,2}{R_2})\\
&=&\gamma^{\mathrm R}(\{(1, 7), (2, 6), (6, 8), (7, 2), (3, 5), (4, 6), (4, 2), (6, 4)\})\\
&=&\{(1, 7),  (2, 6), (6, 8), (7, 2), (3, 5), (4, 6), (4, 2), (6, 4),
\\&&(1,2), (2,8),  (2,4), (4,8), (7,8), (7,4),  (6,2),(1,8),(1,4),(7,6),(1,6)
\\&&(1,1),(2,2),(3,3),(4,4),(5,5),(6,6),(7,7),(8,8),(9,9)\}.
\end{array}
\]
Also,
\[
\begin{array}{rcl}
\gamma(\ShiftDiam{2,4}{\gamma^R(R_1)}\cup\ShiftDiam{0,2}{\gamma^R(R_2)})&=&
\gamma(\{(1,7),(2,6),(6,8),(7,2),(1,2),(2,8),(7,6),(1,6),(7,8),(1,8),\\
&&(1,1),(2,2),(6,6),(7,7),(8,8),(9,9) \}\\
&&\cup \{(3,5),(4,6), (4,2),(6,4),(6,3),(2,2),(3,3),(4,4),(5,5),(6,6) \})\\
&=&\{(1,7),(2,6),(6,8),(7,2),(1,2),(2,8),(7,6),(1,6),(7,8),(1,8),(3,5),(4,6),\\
&& (4,2),(6,4),(6,2),
(2,4),(4,8),(7,4),(1,4),\\
&&(1,1),(2,2),(3,3),(4,4),(5,5),(6,6),(7,7),(8,8),(9,9)\}\\
&=&\gamma^{\mathrm R}(R)
\end{array}
\]
\end{example}
\def\B{\mathbb B}
\subsection{Infinite matrices again}
We consider the set $\mathcal R$ of infinite matrices indexed by $\N\setminus\{0\}$, whose entries are in the boolean semiring $\B$ and which have only a finite number of non zero entries.  Consider the map which sends a set $s\in \mathcal S^{ARef}$ to the matrix $M_s$ such that $M_s[i,j]=1$ is $(i,j)\in s$ and $0$ otherwise. Through this map, $\mathcal S^{ARef}_n$ is in a one to one correspondence with a finite subset $\widetilde{\mathcal S_n}\subset\mathcal R$. More precisely, $$\widetilde{\mathcal S_n}=\{(m_{i,j})_{i,j>0}: m_{i,j}=0\mbox{ if } i>n+1 \mbox{ or }j>n+1\mbox{ or }i=j\}.$$ Furthermore each $\widetilde{\mathcal S_n}$ is stable for the sum and is isomorphic to the monoid $\mathcal S^{ARef}_n$ and we
observe that  the precomposition $\tilde\Diamond$ defined by $\tilde\Diamond(\ShiftBoxnk{i,n})(M)=\mathfrak n_{i,k}M\ ^t\mathfrak n_{i,n}$ is isomorphic to $\Diamond$.
\begin{example}\rm
We have
$$M_{\{(1,4),(2,3),(3,5),(4,2)\}}=\left[\begin{array}{cccccc}0&0&0&1&0&\cdots\\
0&0&1&0&0&\cdots\\
0&0&0&0&1&\cdots\\
0&1&0&0&0&\cdots\\
\vdots&\vdots&\vdots&\vdots&\vdots&
\end{array}\right].$$
And we check that
\[\begin{array}{rcl}
\mathfrak n_{24}M_{\{(1,4),(2,3),(3,5),(4,2)\}}\ ^t\mathfrak n_{24}&=&%
\left[\begin{array}{ccccccccc}0&0&\color{NavyBlue}0&\color{NavyBlue}0&\color{NavyBlue}0&0&1&0&\cdots\\
0&0&\color{NavyBlue}0&\color{NavyBlue}0&\color{NavyBlue}0&1&0&0&\cdots\\
\color{NavyBlue}0&\color{NavyBlue}0&\color{NavyBlue}0&\color{NavyBlue}0&\color{NavyBlue}0&\color{NavyBlue}0&\color{NavyBlue}0&\color{NavyBlue}0&\cdots\\
\color{NavyBlue}0&\color{NavyBlue}0&\color{NavyBlue}0&\color{NavyBlue}0&\color{NavyBlue}0&\color{NavyBlue}0&\color{NavyBlue}0&\color{NavyBlue}0&\cdots\\
\color{NavyBlue}0&\color{NavyBlue}0&\color{NavyBlue}0&\color{NavyBlue}0&\color{NavyBlue}0&\color{NavyBlue}0&\color{NavyBlue}0&\color{NavyBlue}0&\cdots\\
0&0&\color{NavyBlue}0&\color{NavyBlue}0&\color{NavyBlue}0&0&0&1&\cdots\\
0&1&\color{NavyBlue}0&\color{NavyBlue}0&\color{NavyBlue}0&0&0&0&\cdots\\
\vdots&\vdots&\vdots&\vdots&\vdots&\vdots&\vdots&\vdots&
\end{array}\right]=M_{\{(1,7),(2,6),(6,8),(7,2)\}}\\
&=& M_{\ShiftDiamnk{2,4}\{(1,4),(2,3),(3,5),(4,2)\}}.
\end{array}
\]
\end{example}
Each symbol $a_{M}^{(m)}$ can be assimilated with a $(m+1)\times (m+1)$ matrix $M^{(m)}$ which is obtained by considering only the entries $M_{i,j}$ of $M$ such that $i,j\leq m+1$. Hence, the composition $a_{M}^{(n)}\circ_i a_{M'}^{(m)}$ is obtained by summing  two $(m+n)\times (m+n)$ matrices: $\left(n_{i,m}M\ ^t\mathfrak n_{i,m}\right)^{(n+m)}$ and $\left(n_{0,i}M'\ ^t\mathfrak n_{0,i}\right)^{(n+m)}$.
\begin{example}
\rm One has
\[\begin{array}{rcl}
M_{\{(1, 4), (2, 3), (3, 5), (4, 2)\}}^{(5)}\circ_2 M_{\{(2,4),(3,5),(3,1),(5,3)\}}^{(4)}&=&
{\color{WildStrawberry}\left[\begin{array}{cccccc}0&0&0&1&0&0\\0&0&1&0&0&0\\0&0&0&0&1&0\\0&1&0&0&0&0\\0&0&0&0&0&0\\0&0&0&0&0&0 \end{array}\right]}
\circ_2 {\color{NavyBlue}
\left[\begin{array}{ccccc}0&0&0&0&0\\0&0&0&1&0\\1&0&0&0&1\\0&0&0&0&0\\0&0&1&0&0 \end{array}\right]}
\\\\
&=&
\left[\begin{array}{ccccccccc}\color{WildStrawberry}0&\color{WildStrawberry}0&0&0&0&\color{WildStrawberry}0&\color{WildStrawberry}1&\color{WildStrawberry}0&\color{WildStrawberry}0\\
\color{WildStrawberry}0&{\color{WildStrawberry}0}+{\color{NavyBlue}0}&\color{NavyBlue}0&\color{NavyBlue}0&\color{NavyBlue}0&{\color{WildStrawberry}1}+\color{NavyBlue} 0&\color{WildStrawberry}0&\color{WildStrawberry}0&\color{WildStrawberry}0\\
0&\color{NavyBlue}0&\color{NavyBlue}0&\color{NavyBlue}0&\color{NavyBlue}1&\color{NavyBlue}0&0&0&0\\
0&\color{NavyBlue}1&\color{NavyBlue}0&\color{NavyBlue}0&\color{NavyBlue}0&\color{NavyBlue}1&0&0&0\\
0&\color{NavyBlue}0&\color{NavyBlue}0&\color{NavyBlue}0&\color{NavyBlue}0&\color{NavyBlue}0&0&0&0\\
\color{WildStrawberry}0&{\color{WildStrawberry}0}+{\color{NavyBlue}0}&\color{NavyBlue}0&\color{NavyBlue}1&\color{NavyBlue}0&{\color{WildStrawberry}0}+{\color{NavyBlue}0}&\color{WildStrawberry}0&\color{WildStrawberry}1&\color{WildStrawberry}0\\
\color{WildStrawberry}0&\color{WildStrawberry}1&0&0&0&\color{WildStrawberry}0&\color{WildStrawberry}0&\color{WildStrawberry}0&\color{WildStrawberry}0\\
\color{WildStrawberry}0&\color{WildStrawberry}0&0&0&0&\color{WildStrawberry}0&\color{WildStrawberry}0&\color{WildStrawberry}0&\color{WildStrawberry}0\\
\color{WildStrawberry}0&\color{WildStrawberry}0&0&0&0&\color{WildStrawberry}0&\color{WildStrawberry}0&\color{WildStrawberry}0&\color{WildStrawberry}0
\end{array}\right]\\\\
&=&M_{\{(1,7),(2,6),(3,5),(4,2),(4,6),(6,4),(6,8),(7,2)\}}^{(9)}.
\end{array}
\]
This is coherent with Example \ref{ExARAS}.
\end{example}
Set $M^+=M+M^2+\cdots$. Remarking that $M^+_s=M_{\gamma(s)}$, we deduce that
 if $M\in\widetilde{\mathcal S_n}$ then $M^+=M+M^2+\cdots$ still belongs to $\widetilde{\mathcal S_n}$. Let $\widetilde{\mathcal S_n}^+=\{M^+:M\in\widetilde{\mathcal S_n}\}$ and  $\widetilde{\mathcal S}^+=\bigcup_n\widetilde{\mathcal S_n}^+$. The set $\widetilde{\mathcal S}^+$ endowed with the operation $M\oplus M'=(M+M')^+$ is a monoid isomorphic to $\mathrm {QOSet}$. Hence, one easily checks that
\begin{proposition}
 The homomorphism $\widetilde\Diamonddot\in\Hom(\bigbox,\Hom(\widetilde{\mathcal S}^+,\widetilde{\mathcal S}^+))$ defined by $\widetilde\Diamonddot(\ShiftBoxnk{i,k})(M)=\left(\mathfrak n_{i,k}M\ ^t\mathfrak n_{i,k}\right)^+$ is  a precomposition isomorphic to $\Diamonddot$.
\end{proposition}
\section{Action on languages}\label{sec6}
The aim of this section is to describe regular languages by using the operads defined above.
More precisely, we show that the set of regular languages is a module on each of these operads.
Furthermore, we prove that each regular language can be expressed  by an operator acting on symbols or $\emptyset$.
Finally, we show that the operad $\QOSET$ is optimal in the sense that its action is faithful.

\subsection{Action of $\ARef$}
We recall that a grammar $G$ is a $4$-tuple $(\Sigma,\Gamma,S,P)$ where  $\Sigma$ is a terminal alphabet, $\Gamma$ a nonterminal alphabet, $S\in \Gamma$ an axiom and $P$ a set of productions. The set  of productions is a relation that contains couples of the form $X \rightarrow \alpha$, with $X$ in $\Gamma$ and $\alpha$ in $(\Sigma\cup\Gamma)^*$.
We denote by $\Rightarrow$ the catenation stability closure of $\rightarrow$ \emph{i.e.}, the transitive and reflexive closure of the relation $\rightarrow$.
The language denoted by $G$ is the set $\mathbb L(G)=\{w \in \Sigma^*\mid S_1 \Rightarrow w\}$.
\medskip

An $\varepsilon$-automaton is a $5$-tuple $(Q,\Sigma,\delta,i,F)$
where $Q$ is a set of states, $\Sigma$ is an alphabet, $\delta$ is a transition function from $Q \times (\Sigma\cup\{\varepsilon\})$ to $2^{Q}$,  $i\in Q$ is an initial state and $F\subset Q$ is a set of final states.
The transition function $\delta$ can be extended for any integer $n$ as the function $\delta_{\varepsilon,n}$ from $Q$ to $2^Q$ as follows:
for any $q\in Q$, $\delta_{\varepsilon,0}(q)=\delta(q,\varepsilon)$, $\delta_{\varepsilon,n+1}(q)=\bigcup_{q'\in\delta(q,\varepsilon)} \delta_{\varepsilon,n}(q')$.
The transition function $\delta$ can also be extended as the function $\delta'$ from $2^Q\times \Sigma^*$ as follows:
for any $Q\subset \Gamma_n$, for any word $w$ in $\Sigma^*$,
$$\delta'(q,\varepsilon)=\{q\}\cup \left\{q'\mid  q'\in\delta_{\varepsilon,n}(q)\mbox{ for some }n\right\},$$
$$\delta'(Q,w)=\bigcup_{q\in Q,q'\in \delta'(q,w)} \delta'(q',\varepsilon),\
\delta'(q,aw)=\delta'(\delta(q,a),w).$$
The language recognized by a $\varepsilon$-automaton is the set $\{w\mid \delta'(i,w)\cap F\neq\emptyset\}$.\medskip

We associate to each element $a_R^{(n)}\in \ARef_n$  a list of productions $\mathtt P(a_R^{(n)})$ defined by
\begin{enumerate}
\item $\mathtt S_i\rightarrow \mathtt a_{i}\mathtt S_{i+1}$ for each $1\leq i\leq k$,
\item $\mathtt S_{i}\rightarrow\mathtt  S_{i'}$ if $(i,i')\in R$,
\item $\mathtt S_{k+1}\rightarrow \varepsilon$,
\end{enumerate}
and we construct the grammar $\mathbf G_R^{(n)}:=(\mathtt A_n,\Gamma_n,\mathtt S_1,\mathtt P(a_R^{(n)}))$ with $\mathtt A_n:=\{\mathtt a_i:1\leq i\leq n\}$ and $\Gamma_n:=\{\mathtt S_i:1\leq i\leq n+1\}$.
\begin{example}\rm
Let $a^{(5)}_{\{(1,4),(2,3),(3,5),(4,2)\}}$, we have
\[\mathtt P(a^{(5)}_{\{(1,4),(2,3),(3,5),(4,2)\}})
=\left\{
\begin{array}{l}
\mathtt S_1\rightarrow \mathtt a_1\mathtt S_2,\\
\mathtt S_1\rightarrow \mathtt S_4,\\
\mathtt S_2\rightarrow \mathtt a_2\mathtt S_3,\\
\mathtt S_2\rightarrow \mathtt S_3,\\
\mathtt S_3\rightarrow \mathtt a_3\mathtt S_4,\\
\mathtt S_3\rightarrow \mathtt S_5,\\
\mathtt S_4\rightarrow \mathtt a_4\mathtt S_5,\\
\mathtt S_4\rightarrow \mathtt S_2,\\
\mathtt S_5\rightarrow \mathtt a_5\mathtt S_6,\\
\mathtt S_6\rightarrow \varepsilon.
\end{array}
\right.
\]
\end{example}
\medskip
\begin{lemma}\label{LangRat}
The language $\mathbb L(\mathbf G_R^{(n)})$ denoted  by the grammar $\mathbf G_R^{(n)}$ is regular.
\end{lemma}
\begin{proof}
 It is sufficient to remark that $\mathbb L(\mathbf G_R^{(n)})$ is recognized by the $\varepsilon$-automaton $\mathcal A(a_R^{(n)})=(\Gamma_n, \mathtt A_n, \delta_{R}^{(n)}, \mathtt S_1, \{\mathtt S_{n+1}\})$ where the transitions $\delta_{R}^{(n)}$ are
\[
\begin{array}{l}
\displaystyle \mathtt S_i\mathop\rightarrow^{\mathtt a_i}\mathtt S_{i+1}\ \mbox{ for each }1\leq i\leq k,\\
\displaystyle\mathtt S_i\mathop\rightarrow^{\varepsilon}\mathtt S_{j}\ \mbox{ for each }(i,j)\in R.
\end{array}
\]
\end{proof}
\medskip

Note that the automaton $\mathcal A(a_R^{(n)})$ is just an interpretation of the relation $R$ by adding transitions.
\begin{example}\label{Op2Aut}\rm
We obtain the automaton  $\mathcal A(a^{(5)}_{\{(1,4),(2,3),(3,5),(4,2)\}})$ from the graph of the relation $\{(1,4),(2,3),(3,5),(4,2)\}$
%
%
%
%
%
%
%
%
%

\begin{figure}[H]
    \centerline{
      \begin{tikzpicture}[node distance=1.25cm,bend angle=30]
            \node (1) {$1$};
            \node[above of=1] (2) {$2$};
            \node[above right of=2] (3) {$3$};
            \node[right of=3] (4) {$4$};
            \node[below right of=4] (5) {$5$};
            \node[below of=5] (6) {$6$};
            \path[->]
              (1) edge (4)
              (2) edge (3)
              (3) edge (5)
              (4) edge (2)
            ;
      \end{tikzpicture}
    }
\end{figure}

by adding transitions:

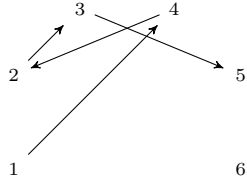
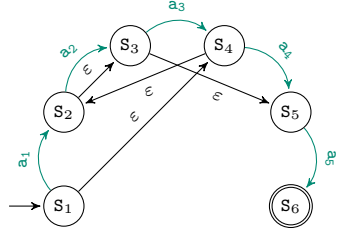
\begin{figure}[H]
    \centerline{
      \begin{tikzpicture}[node distance=1.25cm,bend angle=40]
            \node[state,initial] (1) {$\mathtt{S}_1$};
            \node[state,above of=1] (2) {$\mathtt{S}_2$};
            \node[state,above right of=2] (3) {$\mathtt{S}_3$};
            \node[state,right of=3] (4) {$\mathtt{S}_4$};
            \node[state,below right of=4] (5) {$\mathtt{S}_5$};
            \node[state,below of=5,accepting] (6) {$\mathtt{S}_6$};
            \path[->]
              (1) edge node[above,sloped,pos=0.5] {$\varepsilon$} (4)
              (2) edge node[above,sloped,pos=0.5] {$\varepsilon$} (3)
              (3) edge node[below,sloped,pos=0.6] {$\varepsilon$} (5)
              (4) edge node[below,sloped,pos=0.5] {$\varepsilon$} (2)
            ;
            \path[->,color=PineGreen]
              (1) edge[bend left] node[above,sloped,pos=0.5] {$\mathtt{a_1}$} (2)
              (2) edge[bend left] node[above,sloped,pos=0.5] {$\mathtt{a_2}$} (3)
              (3) edge[bend left] node[above,sloped,pos=0.5] {$\mathtt{a_3}$} (4)
              (4) edge[bend left] node[above,sloped,pos=0.5] {$\mathtt{a_4}$} (5)
              (5) edge[bend left] node[above,sloped,pos=0.5] {$\mathtt{a_5}$} (6)
            ;
      \end{tikzpicture}
    }
\end{figure}

\end{example}
If $L_1,\dots, L_n$ are  languages,  we define $\mathbf G_R^{(n)}(L_1,\dots,L_n)=\mathbb L(\mathbf G_R^{(n)})|_{\mathtt a_i=L_i}$, that is the language $\mathbb L(G_R^{(n)})$ denoted  by the grammar $\mathbf G_R^{(n)}$ where each letter $\mathtt a_i$ is replaced by the language $L_i$.
\begin{example}
\rm
Using the same relation than in Example \ref{Op2Aut} we find
\[
\mathbb L(G_R^{(5)})=(\mathtt a_1+\varepsilon)(\mathtt a_3+\mathtt a_2\mathtt a_3)^*(\mathtt a_5+\mathtt a_2\mathtt a_5+
(\mathtt a_3+\mathtt a_2\mathtt a_3)\mathtt a_4\mathtt a_5)+\mathtt a_4\mathtt a_5.
\]
Therefore, if $L_1,\dots, L_5$ are five languages,
\[
G_R^{(5)}(L_1,\dots,L_5)=(L_1+\varepsilon)(L_3+L_2L_3)^*(L_5+L_2L_5+
(L_3+L_2L_3)L_4L_5)+L_4L_5.
\]

\end{example}
It is easy to see that this construction is compatible with the partial compositions in $\ARef$. Indeed,
\begin{equation}\label{AR2G}
\mathbf G_{R\circ_iR'}^{(n+m-1)}(L_1,\dots,L_{n+m-1})=\mathbf G_{R}^{(n)}(L_1,\dots,L_{i-1},\mathbf G_{R'}^{(m)}(L_i,\dots,L_{i+m-1}),L_{i+m},\dots,L_{n+m-1}),\end{equation}
 for each $a_R^{(n)}\in \ARef_n$,\ $a_{R'}^{(m)}\in \ARef_{m}$, and $i\leq n$.
Indeed,
\[
\mathtt P(a_{R\circ_i R'}^{(n+m-1)})=\left\{
\begin{array}{l}
\mathtt S_j\rightarrow \mathtt a_{j}\mathtt S_{j+1}\mbox{ for each }0\leq j\leq n+m-1,\\
\mathtt S_{\ell}\rightarrow\mathtt  S_{\ell'}\mbox{ if }(\ell,\ell')=\ShiftDiam{i,m}{(j,j')}\mbox{ for }(j,j')\in R,\\
\mathtt S_{\ell}\rightarrow\mathtt  S_{\ell'}\mbox{ if }(\ell,\ell')=\ShiftDiam{0,i}{(j,j')}\mbox{ for }(j,j')\in R',\\
\mathtt S_{n+m}\rightarrow\varepsilon.
\end{array}
\right.
\]
Hence, we have
\[\begin{array}{rcl}
\mathtt P\left(a_{R\circ_i R'}^{(n+m-1)}\right)&=&\{\mathtt S_j\rightarrow \mathtt a_{j}\mathtt S_{j+1}:0\leq j\leq n+m-1\}\\
&& \cup\{\mathtt S_{\ell}\rightarrow\mathtt  S_{\ell'}: \mathtt S_{j}\rightarrow\mathtt  S_{j'}\in \mathtt P(a_R^{(n)})\mbox{ and } (\ell,\ell')=\ShiftDiam{i,m}{(j,j)}\}\\
&& \cup\{\mathtt S_{\ell}\rightarrow\mathtt  S_{\ell'}: \mathtt S_{j}\rightarrow\mathtt  S_{j'}\in \mathtt P(a_{R'}^{(m)})\mbox{ and } (\ell,\ell')=\ShiftDiam{0,i}{(j,j)}\}\\
&&\cup \{ \mathtt S_{n+m-1}\rightarrow\varepsilon\}
\end{array}
\]
We deduce that
\[
\mathbb L(\mathbf G_{R\circ_iR'}^{(n+m-1)})=
\mathbf G_{R}^{(n)}(\mathtt a_1,\dots,\mathtt a_{i-1},\mathbf G_{R'}^{(m)}(\mathtt a_i,\dots,\mathtt a_{i+m-1}),\mathtt a_{i+m},\dots,\mathtt a_{n+m-1}).
\]
This implies (\ref{AR2G}).
\begin{remark}\rm Alternatively, the construction on grammars can be described in terms of automata. The automaton $\mathcal A(a^{(n+m-1)}_{R\circ_i R'})$ is obtained by replacing the transition
$\displaystyle\mathtt S_i\mathop\rightarrow^{\mathtt a_i}\mathtt S_{i+1}$ in $\mathcal A(a^{(n)}_{R})$ by a copy of the automaton $\mathcal A(a^{(m)}_{R'})$ and hence relabeling the vertices and edges.
\end{remark}
\noindent
Setting $a_R^{(n)}.(L_1,\dots,L_n)=\mathbf G_R^{(n)}(L_1,\dots,L_n)$ we define an action of the operad $\ARef$ on languages.
\begin{theorem}\label{AREFmod}
The sets $2^{\Sigma^*}$ and $Reg(\Sigma)$ are $\ARef$-modules.
\end{theorem}
\begin{proof}
The fact that  $2^{\Sigma^*}$ is a $\ARef$-module is a direct consequence of (\ref{AR2G}).\\
Remarking that $\mathbb L(G_R^{(n)})\in Reg(\mathtt A_n)$ (Lemma \ref{LangRat}), we deduce that $\mathbf G_R^{(n)}(L_1,\dots,L_n)\in Reg(\Sigma)$ when $L_1,\dots,L_n\in Reg(\Sigma)$. Equivalently, $Reg(\Sigma)$ is $\ARef$-module.
\end{proof}
\medskip

Note that the action of $\ARef$ can be defined directly from $\mathcal{DT}$.
For any $(a_{T_1}^{(n)},a_{T_2}^{(n)})\in\mathcal{DT}_n$, we construct the grammar
$\mathbf G_{T_1,T_2}^{(n)}:=(\mathtt A_n,\Gamma_n,\mathtt S_1,\mathtt P_{DT}(a_R^{(n)}))$ where the production rules $\mathtt P_{DT}\left(a_R^{(n)}\right)$ are
\begin{enumerate}
\item $\mathtt S_i\rightarrow \mathtt a_{i}\mathtt S_{i+1}$ for each $1\leq i\leq n$,
\item $\mathtt S_i\rightarrow \mathtt S_{i'}$ if $(i',i-1)\in T_2$ or $(i,i'-1)\in T_1$,
\item $\mathtt S_{n+1}\rightarrow \varepsilon$.
\end{enumerate}
\begin{example}\rm
Let $((13)(24)(34),(23))\in\mathcal{DT}_5$. The grammar $\mathbf G_{(13)(22)(34),(23)}^{(5)}$ is
\begin{equation}\label{Gram1}\left\{\begin{array}{rcl}
\mathtt S_1&\rightarrow& \mathtt a_1\mathtt S_2,\\
\mathtt S_1&\rightarrow& \mathtt S_4,\\
\mathtt S_2&\rightarrow &\mathtt a_2\mathtt S_3,\\
\mathtt S_2&\rightarrow &\mathtt S_3,\\
\mathtt S_3&\rightarrow &\mathtt a_3\mathtt S_4,\\
\mathtt S_3&\rightarrow &\mathtt S_5,\\
\mathtt S_4&\rightarrow &\mathtt a_4\mathtt S_5,\\
\mathtt S_4&\rightarrow &\mathtt S_2,\\
\mathtt S_5&\rightarrow &\mathtt a_5\mathtt S_6,\\
\mathtt S_6&\rightarrow&\varepsilon.\end{array}\right.\end{equation}
Note that we recover the grammar $\mathbf G_{\{(1,4),(2,3),(3,5),(4,2)\}}^{(5)}$.
\end{example}
In general we have
\begin{proposition} For each  $(a_{T_1}^{(n)},a_{T_2}^{(n)})\in\mathcal{DT}_n$,
$\mathbf G_{T_1,T_2}^{(n)}=\mathbf G_{\xi(T_1,T_2)}^{(n)}.$\\
Here $\xi$ denotes the morphism from $\mathcal{DT}$ to $\ARef$ as defined in Section \ref{SDT2ARef}.
\end{proposition}

\subsection{Operadic expressions for regular languages}
The following proposition shows that any regular language admits an expression involving an operator of $\ARef$ 	and symbols of the alphabet or $\emptyset$.
\begin{proposition}\label{Reg2Op}
Each regular language $L\in Reg(\Sigma)$ can be written as
\[
L=a_R^{(n)}(\alpha_1,\dots,\alpha_n)
\]
for some $n\geq 1$, $a_R^{(n)}\in\ARef_n$, and $\alpha_1,\dots,\alpha_n\in\{\{\mathtt a\}:\mathtt a\in\Sigma\}\cup\{\emptyset\}$.
\end{proposition}
\begin{proof}
First note that $\{\mathtt a\}=a^{(1)}_{\emptyset}(\{\mathtt a\})$, $\{\varepsilon\}=a^{(1)}_{\{(1,2)\}}(\emptyset)$, and $\emptyset=a^{(1)}_{\emptyset}(\emptyset)$.\\
 Suppose now that $L, L'\in Reg(\Sigma)$ are two regular languages satisfying
\[
L=a_R^{(n)}(\alpha_1,\dots,\alpha_n)\mbox{ and }L'=a_{R'}^{(m)}(\alpha'_1,\dots,\alpha'_{m})
\]
for some $m\geq 1$, $a_R^{(n)}\in\ARef_n$, $a_{R'}^{(m)}\in\ARef_{m}$; and $\alpha_1,\dots,\alpha_n,\alpha'_1,\dots,\alpha'_{m}\in\{\{\mathtt a\}:\mathtt a\in\Sigma\}\cup\{\emptyset\}$. We have
\begin{equation}\label{L+L'}L+L'=a_{R''}^{(n+m+1)}(\alpha_1,\dots,\alpha_n,\emptyset,\alpha'_1,\dots,\alpha'_{m})\ \mbox{ with }\ R''=R\cup \ShiftDiamnk{0,n+1}{R'}\cup \{(1,n+2),(n+1,n+m+2)\},\end{equation}
\begin{equation}\label{LL'}LL'=a_{R''}^{(n+m+1)}(\alpha_1,\dots,\alpha_n,\emptyset,\alpha'_1,\dots,\alpha'_{m})\ \mbox{ with }\ R''=R\cup \ShiftDiamnk{0,n+1}{R'}\cup \{(n+1,n+2)\},\end{equation}
\begin{equation}\label{L^*}L^*=a_{R\cup \{(n+1,1),(1,n+1)\}}^{(n)}(\alpha_1,\dots,\alpha_n).\end{equation}

The property is obtained by a straightforward induction.
\end{proof}
\begin{remark}
\rm Note that in Formula~(\ref{LL'}), the symbol $\emptyset$ is important  for the computation of the catenation. For instance, we have $$\mathtt a^+\mathtt b^+=a_{\{(2,1)\}}^{(1)}(\mathtt a)\cdot a_{\{(2,1)\}}^{(1)}(\mathtt b)=a_{\{(2,1),(2,3),(4,3)\}}^{(3)}(\mathtt a,\emptyset,\mathtt b)\neq
a_{\{(2,1),(3,2)\}}^{(2)}(\mathtt a,\mathtt b)=(\mathtt a^+\mathtt b^+)^+.$$ But in some cases it may be omitted. For instance, $$a_{\{(1,2)\}}^{(1)}(\mathtt a)\cdot a_{\{(1,2)\}}^{(1)}(\mathtt b)=a_{\{(1,2),(2,3)\}}^{(2)}(\mathtt a,\mathtt b)=\varepsilon+\mathtt a+\mathtt b+\mathtt a\mathtt b.$$
\end{remark}
\noindent Let us give some examples. First we illustrate the construction described in the proof of Proposition
\ref{Reg2Op}.
\begin{example}
\rm Consider the languages $L=\mathtt b(\mathtt a\mathtt  b^*)+\mathtt a^*$. We have
$
\{\mathtt a\}=a_{\emptyset}^{(1)}(\mathtt a),\ \{\mathtt b\}=a_{\emptyset}^{(1)}(\mathtt b).
$
So
$$
\mathtt b^*=a_{(2,1),(1,2)}^{(1)}(\mathtt b),\, \mathtt a\mathtt  b^*=a_{(4,3),(3,4)}^{(3)}(\mathtt a,\emptyset,\mathtt b)\mbox{ and }\mathtt b(\mathtt a\mathtt  b^*)=a_{(6,5),(5,6)}^{(5)}(\mathtt b,\emptyset,\mathtt a,\emptyset,\mathtt b).
$$
On the other hand $\mathtt a^*=a_{(2,1),(1,2)}^{(1)}(\mathtt a)$, hence
$$L=a_{(6,5),(5,6),(8,7),(7,8),(6,8),(1,7)}^{(7)}(\mathtt b,\emptyset,\mathtt a,\emptyset,\mathtt b,\emptyset,\mathtt a).
$$
\end{example}
\medskip
Manipulating the relations allows to obtain some languages from others. We give here few constructions.
\begin{example}
\
\rm
\begin{itemize}
\item Consider a language $L=a^{(n)}_R(\alpha_1,\dots,\alpha_n)$ with $R\in\ARef_n$ and $\alpha_i\in\{\{\mathtt a\}:\mathtt a\in\Sigma\}$. We define $R_P:=R\cup \{(i,n+1):1\leq i\leq n\}$. The language $a^{(n)}_{R_P}(\alpha_1,\dots,\alpha_n)$ is the set of the prefixes of $L$.\\
For instance, consider $L=a^{(3)}_{(4,1),(1,4)}(\mathtt a,\mathtt b,\mathtt c)=(\mathtt a\mathtt b\mathtt c)^*$ we have
$L=a^{(3)}_{(4,1),(1,4),(2,4),(3,4)}(\mathtt a,\mathtt b,\mathtt c)=(\mathtt a\mathtt b\mathtt c)^*\{\varepsilon,\mathtt a,\mathtt a\mathtt b\}$.
\item For a more general regular language $L$, Proposition \ref{Reg2Op} implies that there exists $n>0$, $R\in\ARef_n$, and $\alpha_i\in\{\{\mathtt a\}:\mathtt a\in\Sigma\}\cup\{\emptyset\}$ satisfying  $L=a^{(n)}_R(\alpha_1,\dots,\alpha_n)$. An \emph{admissible position} is an integer $1\leq i\leq n+1$ such that there exists a path
$\displaystyle i_1=1\mathop\rightarrow^{\beta_1}i_2\mathop\rightarrow^{\beta_2}i_3\cdots i_{p-1}\mathop\rightarrow^{\beta_p}i_p=i_{n+1}$
in  $\mathcal A(a_R^{(n)})$  with either $\beta_i=\varepsilon$ either $\beta_i=\mathtt a_i$ with $\alpha_i\neq\emptyset$ such that $i_\ell=i$ for some $1\leq\ell\leq p-1$. The set of admissible positions  is denoted by $\mathrm{Adm}(R;\alpha_1,\dots,\alpha_n)$. We define $R_P:=R\cup\{(i,n+1):i\in\mathrm{Adm}(R;\alpha_1,\dots,\alpha_n),i\neq n+1\}$. The language $a^{(n)}_{R_P}(\alpha_1,\dots,\alpha_n)$ is the set of the prefixes of $L$.\\
For instance consider $L=a^{(8)}_{(1,4),(3,6),(6,1),(6,9)}(\mathtt a,\mathtt b,\emptyset,
\mathtt c,\mathtt d,\mathtt a,\emptyset,\mathtt b)$. We have $L=(\mathtt{ab}+\mathtt{cd})^+$,
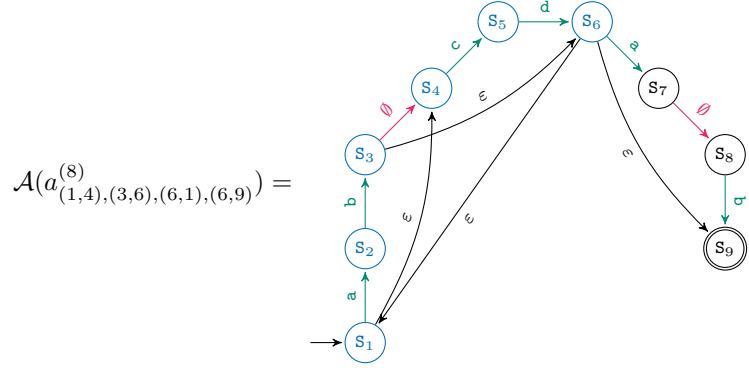
\begin{equation*}
\begin{split}\mathcal A(a_{(1,4),(3,6),(6,1),(6,9)}^{(8)}) =\end{split}\begin{split}
%
      \begin{tikzpicture}[node distance=1.25cm,bend angle=15]
            \node[state,initial,color=NavyBlue] (1) {$\mathtt{S}_1$};
            \node[state,above of=1,color=NavyBlue] (2) {$\mathtt{S}_2$};
            \node[state,above of=2,color=NavyBlue] (3) {$\mathtt{S}_3$};
            \node[state,above right of=3,color=NavyBlue] (4) {$\mathtt{S}_4$};
            \node[state,above right of=4,color=NavyBlue] (5) {$\mathtt{S}_5$};
            \node[state,right of=5,color=NavyBlue] (6) {$\mathtt{S}_6$};
            \node[state,below right of=6] (7) {$\mathtt{S}_7$};
            \node[state,below right of=7] (8) {$\mathtt{S}_8$};
            \node[state,below of=8,accepting] (9) {$\mathtt{S}_9$};
            \path[->]
              (1) edge[bend right] node[above,sloped,pos=0.5] {$\varepsilon$} (4)
              (3) edge[bend right] node[above,sloped,pos=0.5] {$\varepsilon$} (6)
              (6) edge node[below,sloped,pos=0.6] {$\varepsilon$} (1)
              (6) edge[bend right] node[below,sloped,pos=0.5] {$\varepsilon$} (9)
            ;
            \path[->,color=PineGreen]
              (1) edge node[above,sloped,pos=0.5] {$\mathtt{a}$} (2)
              (2) edge node[above,sloped,pos=0.5] {$\mathtt{b}$} (3)
              (4) edge node[above,sloped,pos=0.5] {$\mathtt{c}$} (5)
              (5) edge node[above,sloped,pos=0.5] {$\mathtt{d}$} (6)
              (6) edge node[above,sloped,pos=0.5] {$\mathtt{a}$} (7)
              (8) edge node[above,sloped,pos=0.5] {$\mathtt{b}$} (9)
            ;
            \path[->,color=WildStrawberry]
              (3) edge node[above,sloped,pos=0.5] {$\emptyset$} (4)
              (7) edge node[above,sloped,pos=0.5] {$\emptyset$} (8)
            ;
      \end{tikzpicture}
\end{split} \end{equation*} and $$\mathrm{Adm}(\{(1,4),(3,6),(6,1),(6,9)\};\mathtt a,\mathtt b,\emptyset,
\mathtt c,\mathtt d,\mathtt a,\emptyset,\mathtt b)=\{1,2,3,4,5,6\}.$$
 So $R_P=\{(1,4),(3,6),(6,1),(6,9),(1,9),(2,9),(3,9),(4,9),(5,9)\}$. We check that $a^{(8)}_{R_P}(\mathtt a,\mathtt b,\emptyset,
\mathtt c)=(\mathtt{ab}+\mathtt{cd})^*(\varepsilon+\mathtt a+\mathtt c)=\mathrm{Pref}(L).$
Graphically,
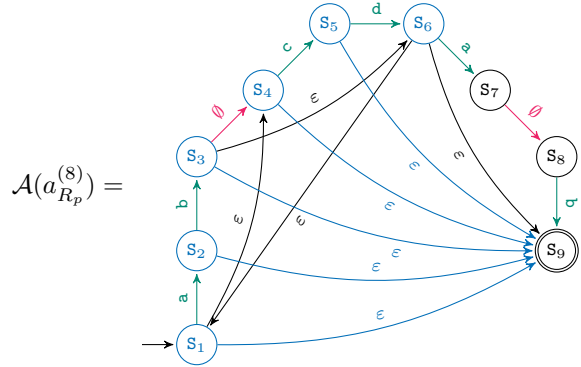
\begin{equation*}
\begin{split}\mathcal A(a_{R_p}^{(8)}) =\end{split}
\begin{split}
%
      \begin{tikzpicture}[node distance=1.25cm,bend angle=15]
            \node[state,initial,color=NavyBlue] (1) {$\mathtt{S}_1$};
            \node[state,above of=1,color=NavyBlue] (2) {$\mathtt{S}_2$};
            \node[state,above of=2,color=NavyBlue] (3) {$\mathtt{S}_3$};
            \node[state,above right of=3,color=NavyBlue] (4) {$\mathtt{S}_4$};
            \node[state,above right of=4,color=NavyBlue] (5) {$\mathtt{S}_5$};
            \node[state,right of=5,color=NavyBlue] (6) {$\mathtt{S}_6$};
            \node[state,below right of=6] (7) {$\mathtt{S}_7$};
            \node[state,below right of=7] (8) {$\mathtt{S}_8$};
            \node[state,below of=8,accepting] (9) {$\mathtt{S}_9$};
            \path[->]
              (1) edge[bend right] node[above,sloped,pos=0.5] {$\varepsilon$} (4)
              (3) edge[bend right] node[above,sloped,pos=0.5] {$\varepsilon$} (6)
              (6) edge node[below,sloped,pos=0.6] {$\varepsilon$} (1)
              (6) edge[bend right] node[below,sloped,pos=0.5] {$\varepsilon$} (9)
            ;
            \path[->,color=PineGreen]
              (1) edge node[above,sloped,pos=0.5] {$\mathtt{a}$} (2)
              (2) edge node[above,sloped,pos=0.5] {$\mathtt{b}$} (3)
              (4) edge node[above,sloped,pos=0.5] {$\mathtt{c}$} (5)
              (5) edge node[above,sloped,pos=0.5] {$\mathtt{d}$} (6)
              (6) edge node[above,sloped,pos=0.5] {$\mathtt{a}$} (7)
              (8) edge node[above,sloped,pos=0.5] {$\mathtt{b}$} (9)
            ;
            \path[->,color=WildStrawberry]
              (3) edge node[above,sloped,pos=0.5] {$\emptyset$} (4)
              (7) edge node[above,sloped,pos=0.5] {$\emptyset$} (8)
            ;
            \path[->,color=NavyBlue]
              (1) edge[bend right] node[above,sloped,pos=0.5] {$\varepsilon$} (9)
              (2) edge[bend right] node[above,sloped,pos=0.5] {$\varepsilon$} (9)
              (3) edge[bend right] node[below,sloped,pos=0.6] {$\varepsilon$} (9)
              (4) edge[bend right] node[below,sloped,pos=0.5] {$\varepsilon$} (9)
              (5) edge[bend right] node[below,sloped,pos=0.5] {$\varepsilon$} (9)
            ;
      \end{tikzpicture}
\end{split} \end{equation*}
Indeed the language recognized by this automaton is $(\mathtt a_1\mathtt a_2\mathtt a_3\mathtt a_4\mathtt a_5+\mathtt a_1\mathtt a_2+\mathtt a_4\mathtt a_5)^*(\varepsilon+\mathtt a_1+\mathtt a_1\mathtt a_2
+\mathtt a_1\mathtt a_2\mathtt a_3+\mathtt a_1\mathtt a_2\mathtt a_3\mathtt a_4+\mathtt a_4 +(\mathtt a_1\mathtt a_2\mathtt a_3\mathtt a_4\mathtt a_5+\mathtt a_1\mathtt a_2+\mathtt a_4\mathtt a_5)(\varepsilon+\mathtt a_6\mathtt a_7\mathtt a_8))$. Setting $\mathtt a_i=\alpha_i$ in this expression, we find $(\mathtt a\mathtt b+\mathtt c\mathtt d)^*(\varepsilon+\mathtt a+\mathtt a\mathtt b
+\mathtt c +(\mathtt a\mathtt b+\mathtt c\mathtt d))=(\mathtt a\mathtt b+\mathtt c\mathtt d)^*(\varepsilon+\mathtt a
+\mathtt c )$ as expected.
\item Symmetrically, the language
 of the suffixes of $L$ is obtained by considering the relation $R_S:=R\cup\{(1,i):i\in\mathrm{Adm}(R;\alpha_1,\dots,\alpha_n),\, i\neq 1\}$. From the example above we obtain $R_S=\{(1,4),(3,6),(6,1),(6,9),(1,2),(1,3),(1,4),(1,5),(1,6)\}$. Graphically,
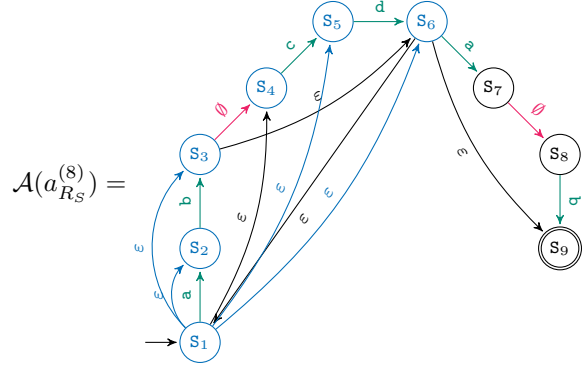
\begin{equation*}
\begin{split}\mathcal A(a_{R_S}^{(8)}) =\end{split}
\begin{split}
%
      \begin{tikzpicture}[node distance=1.25cm,bend angle=15]
            \node[state,initial,color=NavyBlue] (1) {$\mathtt{S}_1$};
            \node[state,above of=1,color=NavyBlue] (2) {$\mathtt{S}_2$};
            \node[state,above of=2,color=NavyBlue] (3) {$\mathtt{S}_3$};
            \node[state,above right of=3,color=NavyBlue] (4) {$\mathtt{S}_4$};
            \node[state,above right of=4,color=NavyBlue] (5) {$\mathtt{S}_5$};
            \node[state,right of=5,color=NavyBlue] (6) {$\mathtt{S}_6$};
            \node[state,below right of=6] (7) {$\mathtt{S}_7$};
            \node[state,below right of=7] (8) {$\mathtt{S}_8$};
            \node[state,below of=8,accepting] (9) {$\mathtt{S}_9$};
            \path[->]
              (1) edge[bend right] node[above,sloped,pos=0.5] {$\varepsilon$} (4)
              (3) edge[bend right] node[above,sloped,pos=0.5] {$\varepsilon$} (6)
              (6) edge node[below,sloped,pos=0.6] {$\varepsilon$} (1)
              (6) edge[bend right] node[below,sloped,pos=0.5] {$\varepsilon$} (9)
            ;
            \path[->,color=PineGreen]
              (1) edge node[above,sloped,pos=0.5] {$\mathtt{a}$} (2)
              (2) edge node[above,sloped,pos=0.5] {$\mathtt{b}$} (3)
              (4) edge node[above,sloped,pos=0.5] {$\mathtt{c}$} (5)
              (5) edge node[above,sloped,pos=0.5] {$\mathtt{d}$} (6)
              (6) edge node[above,sloped,pos=0.5] {$\mathtt{a}$} (7)
              (8) edge node[above,sloped,pos=0.5] {$\mathtt{b}$} (9)
            ;
            \path[->,color=WildStrawberry]
              (3) edge node[above,sloped,pos=0.5] {$\emptyset$} (4)
              (7) edge node[above,sloped,pos=0.5] {$\emptyset$} (8)
            ;

            \path[->,color=NavyBlue]
              (1) edge[bend left,in=135,out=45] node[above,sloped,pos=0.5] {$\varepsilon$} (2)
              (1) edge[bend left,in=135,out=45] node[above,sloped,pos=0.5] {$\varepsilon$} (3)
              (1) edge[bend right] node[above,sloped,pos=0.5] {$\varepsilon$} (5)
              (1) edge[bend right] node[above,sloped,pos=0.5] {$\varepsilon$} (6)
            ;
      \end{tikzpicture}
\end{split} \end{equation*}
\item The language of the factors of $L$ is obtained by first computing the prefixes and then the suffixes. Applying this construction to $L=a^{(8)}_{(1,4),(3,6),(6,1),(6,9)}(\mathtt a,\mathtt b,\emptyset,
\mathtt c,\mathtt d,\mathtt a,\emptyset,\mathtt b)$, we find that the set of the factors of $L$ is denoted by $a^{(8)}_{R_F}(\mathtt a,\mathtt b,\emptyset,
\mathtt c,\mathtt d,\mathtt a,\emptyset,\mathtt b)$ with $$R_F=\{(1,4),(3,6),(6,1),(6,9),(1,9),(2,9),(3,9),(4,9),(5,9),(1,2),(1,3),(1,4),(1,5),(1,6)\}.$$
\item The subwords of $L$ are denoted by the expressions $a_{S}^{(n)}(\alpha_1,\dots,\alpha_n)$ where
$S=R\cup\{(i,i+1):\alpha_i\neq\emptyset\}$. Applying the construction to  $L=a^{(8)}_{(1,4),(3,6),(6,1),(6,9)}(\mathtt a,\mathtt b,\emptyset,
\mathtt c,\mathtt d,\mathtt a,\emptyset,\mathtt b)$, the language of the subwords of $L$ is $a^{(8)}_{(1,4),(3,6),(6,1),(6,9),(1,2),(2,3),(4,5),(5,6),(6,7),(8,9)}(\mathtt a,\mathtt b,\emptyset,
\mathtt c,\mathtt d,\mathtt a,\emptyset,\mathtt b)$. The associated automaton is
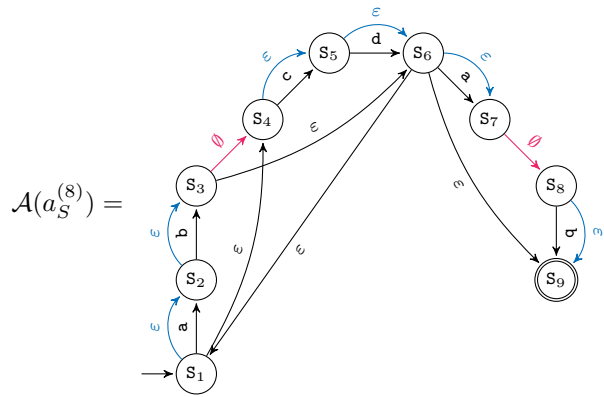
\begin{equation*}
\begin{split}\mathcal A(a_{S}^{(8)}) =\end{split}
\begin{split}
      \begin{tikzpicture}[node distance=1.25cm,bend angle=15]
            \node[state,initial] (1) {$\mathtt{S}_1$};
            \node[state,above of=1] (2) {$\mathtt{S}_2$};
            \node[state,above of=2] (3) {$\mathtt{S}_3$};
            \node[state,above right of=3] (4) {$\mathtt{S}_4$};
            \node[state,above right of=4] (5) {$\mathtt{S}_5$};
            \node[state,right of=5] (6) {$\mathtt{S}_6$};
            \node[state,below right of=6] (7) {$\mathtt{S}_7$};
            \node[state,below right of=7] (8) {$\mathtt{S}_8$};
            \node[state,below of=8,accepting] (9) {$\mathtt{S}_9$};
            \path[->]
              (1) edge[bend right] node[above,sloped,pos=0.5] {$\varepsilon$} (4)
              (3) edge[bend right] node[above,sloped,pos=0.5] {$\varepsilon$} (6)
              (6) edge node[below,sloped,pos=0.6] {$\varepsilon$} (1)
              (6) edge[bend right] node[below,sloped,pos=0.5] {$\varepsilon$} (9)
            ;
            \path[->]
              (1) edge node[above,sloped,pos=0.5] {$\mathtt{a}$} (2)
              (2) edge node[above,sloped,pos=0.5] {$\mathtt{b}$} (3)
              (4) edge node[above,sloped,pos=0.5] {$\mathtt{c}$} (5)
              (5) edge node[above,sloped,pos=0.5] {$\mathtt{d}$} (6)
              (6) edge node[above,sloped,pos=0.5] {$\mathtt{a}$} (7)
              (8) edge node[above,sloped,pos=0.5] {$\mathtt{b}$} (9)
            ;
            \path[->,color=WildStrawberry]
              (3) edge node[above,sloped,pos=0.5] {$\emptyset$} (4)
              (7) edge node[above,sloped,pos=0.5] {$\emptyset$} (8)
            ;
            \path[->,color=NavyBlue]
              (1) edge[bend left,in=135,out=45] node[above,sloped,pos=0.5] {$\varepsilon$} (2)
              (2) edge[bend left,in=135,out=45] node[above,sloped,pos=0.5] {$\varepsilon$} (3)
              (4) edge[bend left,in=135,out=45] node[above,sloped,pos=0.5] {$\varepsilon$} (5)
              (5) edge[bend left,in=135,out=45] node[above,sloped,pos=0.5] {$\varepsilon$} (6)
              (6) edge[bend left,in=135,out=45] node[above,sloped,pos=0.5] {$\varepsilon$} (7)
              (8) edge[bend left,in=135,out=45] node[above,sloped,pos=0.5] {$\varepsilon$} (9)
            ;
      \end{tikzpicture}
\end{split} \end{equation*}
\item The mirror image of $L$ is obtained by computing $a_M^{(n)}(\alpha_n,\dots,\alpha_1)$ where $M=\{(n+2-j,n+2-i):(i,j)\in R\}$. Let us again illustrate the construction on $L=a^{(8)}_{(1,4),(3,6),(6,1),(6,9)}(\mathtt a,\mathtt b,\emptyset,
\mathtt c,\mathtt d,\mathtt a,\emptyset,\mathtt b)$. The mirror image of $L$ is
$a^{(8)}_{(1,4),(3,6),(6,1),(6,9)}(\mathtt b,\emptyset,\mathtt a,
\mathtt d,\mathtt c,\emptyset,\mathtt b,\mathtt a)$. Graphically,
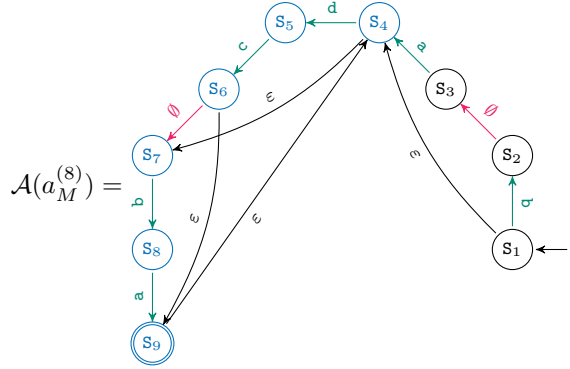
\begin{equation*}
\begin{split}\mathcal A(a_{M}^{(8)}) =\end{split}
\begin{split}
      \begin{tikzpicture}[node distance=1.25cm,bend angle=15]
            \node[state,color=NavyBlue,accepting] (1) {$\mathtt{S}_9$};
            \node[state,above of=1,color=NavyBlue] (2) {$\mathtt{S}_8$};
            \node[state,above of=2,color=NavyBlue] (3) {$\mathtt{S}_7$};
            \node[state,above right of=3,color=NavyBlue] (4) {$\mathtt{S}_6$};
            \node[state,above right of=4,color=NavyBlue] (5) {$\mathtt{S}_5$};
            \node[state,right of=5,color=NavyBlue] (6) {$\mathtt{S}_4$};
            \node[state,below right of=6] (7) {$\mathtt{S}_3$};
            \node[state,below right of=7] (8) {$\mathtt{S}_2$};
            \node[state,below of=8,initial right] (9) {$\mathtt{S}_1$};
            \path[<-]
              (1) edge[bend right] node[above,sloped,pos=0.5] {$\varepsilon$} (4)
              (3) edge[bend right] node[above,sloped,pos=0.5] {$\varepsilon$} (6)
              (6) edge node[below,sloped,pos=0.6] {$\varepsilon$} (1)
              (6) edge[bend right] node[below,sloped,pos=0.5] {$\varepsilon$} (9)
            ;
            \path[<-,color=PineGreen]
              (1) edge node[above,sloped,pos=0.5] {$\mathtt{a}$} (2)
              (2) edge node[above,sloped,pos=0.5] {$\mathtt{b}$} (3)
              (4) edge node[above,sloped,pos=0.5] {$\mathtt{c}$} (5)
              (5) edge node[above,sloped,pos=0.5] {$\mathtt{d}$} (6)
              (6) edge node[above,sloped,pos=0.5] {$\mathtt{a}$} (7)
              (8) edge node[above,sloped,pos=0.5] {$\mathtt{b}$} (9)
            ;
            \path[<-,color=WildStrawberry]
              (3) edge node[above,sloped,pos=0.5] {$\emptyset$} (4)
              (7) edge node[above,sloped,pos=0.5] {$\emptyset$} (8)
            ;
      \end{tikzpicture}
\end{split} \end{equation*}
 The language recognized by $\mathcal A(a_{M}^{(8)})$ is $(\varepsilon+\mathtt a_1\mathtt a_2\mathtt a_3)(\mathtt a_4\mathtt a_5(\varepsilon+\mathtt a_6\mathtt a_7\mathtt a_8)+\mathtt a_7\mathtt a_8)^+$. Specializing to $\mathtt a_1=\mathtt b$, $\mathtt a_2=\emptyset$, $\mathtt a_3=\mathtt a$, $\mathtt a_4=\mathtt d$, $\mathtt a_5=\mathtt c$, $\mathtt a_6=\emptyset$, $\mathtt a_7=\mathtt b$, and $\mathtt a_8=\mathtt a$, we recover the language $(\mathtt d\mathtt c+\mathtt b\mathtt a)^+$ that is the mirror image of $L$.
\end{itemize}
\end{example}
Some more examples:
\begin{example}\rm  Let $\mathtt a_1,\dots,\mathtt a_n$ be $n$ letters. We have
\begin{itemize}
\item $a^{(n)}_{\{(n+1,1),(1,n+1)\}}(\mathtt a_1,\dots,\mathtt a_n)=(\mathtt a_1\cdots\mathtt a_n)^*$.
\item $a^{(n)}_{\{(i,j):i\neq j\}}(\mathtt a_1,\dots,\mathtt a_n)=(\mathtt a_1+\cdots+\mathtt a_n)^*$.
\item $a^{(n)}_{\{(n+1,1)\}\cup  \{(i,n+1):1\leq i\leq n\}}(\mathtt a_1,\dots,\mathtt a_n)=(\mathtt a_1+\mathtt a_1\mathtt a_2+\cdots+\mathtt a_1\cdots\mathtt a_n)^*$.
\item $a^{(n)}_{\{(n+1,1)\}\cup  \{(1,i+1):1\leq i\leq n\}}(\mathtt a_1,\dots,\mathtt a_n)=(\mathtt a_n+\mathtt a_{n-1}\mathtt a_n+\cdots+\mathtt a_1\cdots\mathtt a_n)^*$.
\item
$
a^{(n)}_{\{(i+1,i):1\leq i\leq n\}}(\mathtt a_1,\dots,\mathtt a_n)=\{w\in\{\mathtt a_1,\dots,\mathtt a_n\}^*:w=\mathtt a_1w'\mathtt a_n\mbox{ and }w=u\mathtt a_i\mathtt a_jv\mbox{ implies }j\leq i+1\}.
$
\end{itemize}
\end{example}
\medskip
\subsection{Action of $\QOSET$} \label{sub:action_qoset}
Let $a_{R}^{(n)}\in\ARef_n$. If we compare the grammars $G_R^{(n)}$ and $G_{\gamma R}^{(n)}$ ($\gamma R$ being the transitive cloture of $R$), we observe that $S_i\rightarrow S_\ell \in P(a_{\tilde\gamma R}^{(n)})$ implies there exists $i_1=i,i_2,\dots,i_p=\ell$ such that $S_{i_h}\rightarrow S_{h+1} \in P(a_{ R}^{(n)})$ for each $1\leq h<\ell$. Hence, the languages $\mathbb L(G_R^{(n)})$ and $\mathbb L(G_{\gamma R}^{(n)})$ are equal.
\begin{example}\rm
Consider $R=\{(1,2),(2,3)\}$, we have $\tilde\gamma(R)=\{(1,2),(2,3),(1,3)\}$. We have
\[
P(a_{ R}^{(2)})=\left\{\begin{array}{l}\mathtt S_1\rightarrow \mathtt a_1\mathtt S_2,\\
\mathtt S_1\rightarrow \mathtt S_2,\\
\mathtt S_2\rightarrow \mathtt a_2\mathtt S_3,\\
\mathtt S_2\rightarrow \mathtt S_3,\\
\mathtt S_3\rightarrow\varepsilon,
\end{array}\right.
\mbox{ and }
P(a_{ \tilde\gamma R}^{(2)})=\left\{\begin{array}{l}\mathtt S_1\rightarrow \mathtt a_1\mathtt S_2,\\
\mathtt S_1\rightarrow \mathtt S_2,\\
\mathtt S_1\rightarrow \mathtt S_3,\\
\mathtt S_2\rightarrow \mathtt a_2\mathtt S_3,\\
\mathtt S_2\rightarrow \mathtt S_3,\\
\mathtt S_3\rightarrow\varepsilon.
\end{array}\right.
\]
Hence, $\mathbb L(G_R^{(n)})=\{\varepsilon,\mathtt a_1,\mathtt a_1\mathtt a_2,\mathtt a_2\}=\mathbb L(G_{\gamma R}^{(n)})$.
\end{example}
This allows to consider the action of $\mathrm{\OP}(\Diamonddot)$ defined by $a^{(n)}_{\left[R\right]}(L_1,\dots,L_n):=a^{(n)}_{R}(L_1,\dots,L_n)$.\\
Alternatively, the action of $\QOSET$ is defined by $Q(L_1,\dots,L_n)=a^{(n)}_{Q\setminus\Delta}(L_1,\dots,L_n)$. Observing that the operads $\QOSET$ and $\mathrm{\OP}(\Diamonddot)$ are isomorphic and that the isomorphism $\eta$ satisfies $\eta(Q)(L_1,\dots,L_n)=a^{(n)}_{\left[Q\setminus\Delta\right]}(L_1,\dots,L_n)=a^{(n)}_{Q\setminus\Delta}(L_1,\dots,L_n)=Q(L_1,\dots,L_n)$, the action of $\QOSET$ is compatible with the partial compositions. Hence, Theorem \ref{AREFmod} implies
\begin{corollary}
The sets $2^{\Sigma^*}$ and $Reg(\Sigma)$ are $\QOSET$-modules.
\end{corollary}
Now, we prove that the operad $\QOSET$ is optimal in the sense that two different operators act in two different ways on regular languages. That is:
\begin{theorem}\label{QOSetmodule}
If $\Sigma$ is an alphabet with at least two letters then
$Reg(\Sigma)$ is a faithful $\QOSET$-module.
\end{theorem}
\begin{proof}
Let $Q_1\neq Q_2\in \QOSET_n$ be two quasiorders. Without loss of generality, we suppose that there exists $(i,j)\in Q_1$ such that $(i,j)\not\in Q_2$. Let $\Sigma_n=\{\mathtt a_1,\dots,\mathtt a_n\}$ be an alphabet. The constructions above shows that the word $\mathtt a_1\mathtt a_2\dots\mathtt  a_{i-1}\mathtt a_{j}\mathtt a_{j+1}\dots\mathtt  a_{n}$ belongs to $Q_1(\{\mathtt a_1\},\{\mathtt a_2\},\dots,\{\mathtt a_{n}\})$ but not to $Q_2(\{\mathtt a_1\},\{\mathtt a_2\},\dots,\{\mathtt a_{n}\})$. Setting $\mathtt a_\ell=\mathtt a^{\ell-1}\mathtt b$ for each $\ell>0$, this shows the result for an alphabet of size at least $2$.
\end{proof}
\medskip

Note that the number of elements of $\QOSET_n$ is known up to $n=17$ (see \cite{Sloane} {\tt sequence A000798}):
\[
 4, 29, 355, 6942, 209527, 9535241, 642779354, 63260289423, \dots
\]
\begin{example}\
\rm
\begin{itemize}
\item Let us examine the four operators of $\QOSET_1$:
\[
Q_1=\{(1,1),(2,2)\}, Q_2=\{(1,1),(1,2),(2,2)\}, Q_3=\{(1,1),(2,1),(2,2)\}, Q_4=\{(1,1),(1,2),(2,1),(2,2)\},
\]
The four languages are $Q_1(\mathtt  a_1)=\mathtt a_1$, $Q_2( \mathtt a_1)=\varepsilon+\mathtt a_1$, $Q_3=\mathtt a_1^+(=\mathtt  a_1\mathtt  a_1^*)$, and $Q_4=\mathtt a_1^*$.
\item Let us examine the 29 operators of $\QOSET_2$:
{\footnotesize
\[
\begin{array}{|c|c|c|c|c|c|}
\hline
Q\setminus\Delta&Q(\mathtt a_1, \mathtt a_2)&Q\setminus\Delta&Q(\mathtt a_1, \mathtt a_2)&Q\setminus\Delta&Q(\mathtt a_1, \mathtt a_2)\\\hline
\emptyset& \mathtt a_1\mathtt a_2&\{(1,2)\}&\mathtt a_2+\mathtt a_1\mathtt a_2
&\{(1,3)\}&\varepsilon+\mathtt a_1\mathtt a_2\\
\{(2,3)\}&\mathtt a_1+\mathtt a_1\mathtt a_2&
\{(2,1)\}&\mathtt a_1^+\mathtt a_2&\{(3,1)\}&(\mathtt a_1\mathtt a_2)^+\\\{(3,2)\}&\mathtt a_1\mathtt a_2^+&\{(1,2),(2,1)\}&\mathtt a_1^*\mathtt a_2&
\{(1,3),(3,1)\}& (\mathtt a_1\mathtt a_2)^*\\
\{(2,3),(3,2)\}&\mathtt a_1\mathtt a_2^*&
\{(1,2),(3,2)\}&(\varepsilon+\mathtt a_1)\mathtt a_2^+&
\{(2,1),(2,3)\}&a_1^+(\varepsilon+\mathtt a_2)\\
\{(1,3),(2,3)\}&(\varepsilon+\mathtt a_1+\mathtt a_2)&
\{(3,1),(3,2)\}&(\mathtt a_1\mathtt a_2^+)^+&
\{(3,1),(2,1)\}&(\mathtt a_1^+\mathtt a_2)^+\\
\{(1,3),(1,2)\}&\varepsilon+ \mathtt a_2+\mathtt a_1\mathtt a_2&
\{(1,2),(2,3),(1,3)\}&\varepsilon+ \mathtt a_1+\mathtt a_2+\mathtt a_1\mathtt a_2
& \{(2,1),(3,2),(3,1)\}&(\mathtt a_1^+\mathtt a_2^+)^+\\
\{(1,3),(3,2),(1,2)\}&(\varepsilon+\mathtt a_1)\mathtt a_2^++\varepsilon&\{(3,1),(2,3),(2,1)\}&(\mathtt a_1^+(\varepsilon+\mathtt a_2))^+&
\{(2,1),(1,3),(2,3)\}&\varepsilon+ \mathtt a_1^+(\varepsilon+\mathtt a_2)\\
\{(1,2),(3,1),(3,2)\}&((\mathtt a_1+\varepsilon)\mathtt a_2^+)^+&&&&\\\hline
\end{array}
\]}
\[
\begin{array}{|c|c|c|c|}
\hline
Q\setminus\Delta&Q(\{\mathtt a_1\}, \{\mathtt a_2\})&Q\setminus\Delta&Q(\{\mathtt a_1\}, \{\mathtt a_2\})\\\hline
\{(1,2),(2,1),(2,3),(1,3)\}&\varepsilon+\mathtt a_1^*\mathtt a_2&
\{(1,2),(2,1),(3,2),(3,1)\}&(\mathtt a_1^*\mathtt a_2^+)^+\\
\{(1,3),(3,1),(1,2),(3,2)\}&((\varepsilon+\mathtt a_1)\mathtt a_2^+)^*&
\{(1,3),(3,1),(2,1),(2,3)\}&(\mathtt a_1^2(\varepsilon+\mathtt a_2))^*\\
\{(2,3),(3,2),(2,1),(3,1)\}&(\mathtt a_1^+\mathtt a_2^*)^+&
\{(2,3),(3,2),(1,2),(1,3)\}&(\varepsilon+\mathtt a_1)\mathtt a_2^*\\
\{(1,2),(1,3),(2,3),(2,1),(2,3),(3,1)\}&(\mathtt a_1+\mathtt a_2)^*&&\\\hline
\end{array}\]
We illustrate the proof of Theorem \ref{QOSetmodule}. Remarking that $(3,2)\in\{(2,3),(3,2),(2,1),(3,1)\}$, $(3,2)\not\in\{(2,1),(1,3),(2,3)\}$, we have $\mathtt a_1\mathtt a_2\mathtt a_2\in (\mathtt a_1^+\mathtt a_2^*)^+=
\{(2,3),(3,2),(2,1),(3,1)\}(\mathtt a_1,\mathtt a_2)$ and
$\mathtt a_1\mathtt a_2\mathtt a_2\not\in \varepsilon+ \mathtt a_1^+(\varepsilon+\mathtt a_2)=
\{(2,1),(1,3),(2,3)\}(\mathtt a_1,\mathtt a_2)$.
\end{itemize}
\end{example}

\subsection{Back to (simple) multi-tildes}
The purpose of this section is to show that the restriction of the action to (simple) multi-tildes is compatible with the action described in \cite{LMN12}. In this paper, the action of multi-tildes involve another operad: the operad of sets of boolean vectors $\mathcal B=\bigcup_n\mathcal B_n$ with $\mathcal B_n=2^{\mathbb B^n}$ and $\mathbb B=\{0,1\}$. The composition is defined by
\[
E\circ_i F=\{[e_1,\dots,e_{i-1},e_if_1,\dots,e_if_{m},e_{i+1},\dots,e_n]:[e_1,\dots,e_m]\in E, [f_1,\dots,f_m]\in F\}
\]
for $E\in\mathcal B_n$ and $F\in\mathcal B_m$. The action on the languages is defined by
\[
E(L_1,\dots,L_n)=\bigcup_{[e_1,\dots,e_n]\in E}L_1^{e_1}\cdots L_n^{e_n}.
\]
We denote $[x,z]=\{y:x\leq y\leq z\}$. For each $T\in\mathcal T_n$ we set $\mathcal F(T)=\{S\subset T: (x,y),(z,t)\in S\mbox{ implies }[x,y]\cup[z,t]=\emptyset\}$. Finally we define $V(T)=\{v(S):S\in \mathcal F(T)\}$ with $v(S)=(v_1,\dots,v_n)$ where $v_j=0$ if $j\in\bigcup_{(x,y)\in S}[x,y]$ and $1$ otherwise. In \cite{LMN12} we proved that $V$ is an operadic morphism and defined the action $T(L_1,\dots,L_n)=V(T)(L_1,\dots,L_n)$.\\
Remark that $\mathcal T$ is isomorphic to the suboperad of $\mathcal{DT}$ generated by $(a_{T}^{(n)}, a^{(n)}_\emptyset)$ (the isomorphism sends each $T$ to $(a_{T}^{(n)}, a^{(n)}_\emptyset)$.
So we have to prove that $T(L_1,\dots,L_n)=(a_{T}^{(n)}, a^{(n)}_\emptyset)(L_1,\dots,L_n)$. Equivalently,
\[
T(\mathtt a_1,\dots,\mathtt a_n)=\mathbb L(\mathbf G_{T,\emptyset})(\mathtt a_1,\dots, \mathtt a_n).
\]
To this aim, we associate a set of boolean vectors to each grammar $\mathbf G_{T,\emptyset}$ in the following way: we consider the grammar $\mathbf G_{0,1}(T)$ which is obtained from $\mathbf G_{T,\emptyset}$ by substituting to each rule $\mathtt S_i\rightarrow \mathtt a_i\mathtt S_{i+1}$ the rule $\mathtt S_i\rightarrow 1\mathtt S_{i+1}$ and to each rule $\mathtt S_i\rightarrow \mathtt S_{j}$ the rule $\mathtt S_i\rightarrow 0^{j-i}\mathtt S_{j}$. Denote $\mathbb L_{0,1}(T)=\mathbb L(\mathbf G_{0,1}(T))$. Each word of $\mathbb L_{0,1}(T)$ has a length equal to $n$. Remark that
\[
\mathbb L(\mathbf G_{T,\emptyset})(\mathtt a_1,\dots, \mathtt a_n)=\{\mathtt a_1^{e_1}\cdots \mathtt a_n^{e_n}:e_1\dots e_n\in \mathbb L_{0,1}(T)\}.
\]
Assimilating each word $e_1\dots e_n\in \mathbb L_{0,1}(T)$ to the boolean vector $(e_1,\dots,e_n)$ we prove the following result:
\begin{proposition} For any $a_T^{(n)}\in\mathcal T_n$, we have
$a_T^{(n)}(L_1,\dots,L_n)=(a_{T}^{(n)}, a^{(n)}_\emptyset)(L_1,\dots,L_n)$.
\end{proposition}
\begin{proof}
Let us first recall that a \emph{closed} multi-tilde is a multi-tilde $T$ satisfying
\[
(i,j),\, (j+1,\ell)\in T\Rightarrow (i,\ell)\in T.
\]
The \emph{normal form} $\widetilde T$ of a multi-tilde $T$ is the smallest closed multi-tilde containing $T$ as a subset (see \emph{e.g.}\cite{CCM11a}). From the definition of the action of $\mathcal T$, we have $a_{T}^{(n)}(L_1,\dots,L_n)=a_{\widetilde T}^{(n)}(L_1,\dots,L_n)$. From the construction of $\mathbf G_{T,\emptyset}$ we observe that $ \mathbb L_{0,1}(\tilde T)= \mathbb L_{0,1}(T)$. Indeed, it is sufficient to remark that one can add the rule $\mathtt S_i\rightarrow 0^{l-i}\mathtt S_i$ in $\mathbf G_{0,1}(T)$, when $\mathtt S_i\rightarrow 0^{j-i}\mathtt S_j$ and $\mathtt S_j\rightarrow 0^{l-j}$ are two rules of $\mathbf G_{0,1}(T)$, without modifying the language.\smallskip

Thus, we have to prove $a_T^{(n)}(L_1,\dots,L_n)=(a_{T}^{(n)}, a^{(n)}_\emptyset)(L_1,\dots,L_n)$ for any closed multi-tilde $T$. That is $v=0^{i_1}10^{i_2}\cdots10^{i_p}\in V(T)$ (considering the vector as a word) if and only if $v\in \mathbb L_{0,1}(T)$. The case when $p=1$ means that $v=0^{i_1}=0^n$. For convenience, we set $i_0=1$. Obviously $(i_0,i_1),(i_0+i_1+1,i_0+i_1+i_2+1),\dots,(i_0+i_1+\cdots+i_{\ell-1}+2(\ell-1)+1,i_0+i_1+\cdots+i_{\ell}+2(\ell-1)+1)\in T$ if and only if $\displaystyle \mathtt S_1\mathop\rightarrow^*0^{i_1}1\dots 10^{i_\ell}\mathtt S_{i_0+i_1+\cdots+i_\ell+2\ell}$ for any $0\leq \ell\leq p$ (here $\displaystyle E\mathop\rightarrow^*w$ means that we can produce the word $w$ from $E$ by applying a finite sequence of rules). Equivalently $v\in V(T)$ if and only if $\mathtt S_1\displaystyle\mathop\rightarrow^*v\mathtt S_{n+1}\rightarrow v$. This proves the result.
\end{proof}

\section*{Conclusion and perspectives}
\def\rk{\mathrm{rank}}
We have described a faithful action of a combinatorial operad on regular languages. This means that we describe countable operations providing a new kind of expressions for denoting regular languages. One of the interest of the construction is that we propose expressions which are close to the representation by automata. The obtained expressions are more expressive in the sense that most of the complexity of the denoted language is concentrated at the operator. So this allows to define several measures of the complexity of a language. For instance, let us define $\rk_w(L)=\min\{k:\exists Q\in\QOSET_k,\ \alpha_1,\dots,\alpha_k\in \Sigma\cup\{\emptyset\} \mbox{ such that } L=Q(\alpha_1,\dots,\alpha_k)\}$ and
$\rk_h(L)=\min\{h:\exists k\geq 1, O\in\mathcal{DT}_k,\ \alpha_1,\dots,\alpha_k\in \Sigma\cup\{\emptyset\} \mbox{ such that } L=O(\alpha_1,\dots,\alpha_k)\mbox{ and }\#O=h\}$. The two ranks $\rk_w$ and $\rk_h$ can be respectively interpreted as the width and the height of a language. The first one, $\rk_w$, is the minimal number of occurrences of symbols or $\emptyset$ in the expression. The rank  $\rk_h$ expresses the minimal complexity of an operator involved for denoting the languages. These measures will be investigated; in particular a parallel with the size of a minimal (in terms of states or transitions) automaton should be established.
\smallskip

The operads considered in this paper are $\mathtt{SET}$-operads, that are operads that can be constructed from the category $\mathtt{SET}$. We can also consider linear combinations of operators which consists to use $\mathtt{VECT}$-operads based on the category of the vector spaces. By this way, we guess that the infinite matrices studied in our paper are  good candidates to describe a weighted analogue of multi-tilde operators for rational series.
\smallskip

Another perspective is the extension of the conversion methods from automata to expressions using double multi-tildes.
These conversions  were studied in~\cite{CCM10} and in~\cite{CCM12}.
By slightly modifying the action of our operads, we aim to extend these algorithms of conversions, and conversely from expressions to automata \emph{e.g.}, the position functions~\cite{Glu61} or the expression derivatives~\cite{Ant96,Brzo64}.
\smallskip

A last perspective, suggested by the referee, is the following. By the
Alexandroff correspondence~\cite{Ale37}, quasiorders on finite sets are 
in bijection with finite topologies. The question consists in 
investigating if the action of the operad of quasiorders $\QOSET$ on 
languages (see Section~\ref{sub:action_qoset}) has a topological 
interpretation.


\end{document}